\documentclass[11pt]{article}
\usepackage{amssymb,amsthm,amsmath,float}

\usepackage[total={6.5in,8.75in}, top=1.2in, left=0.9in, includefoot]{geometry}
\usepackage{graphicx}
\usepackage[pdftex,bookmarks, colorlinks, citecolor =blue, breaklinks]{hyperref}

\usepackage{url}
\newcommand{\Eq}[1]{(\ref{eq:#1})}
\newcommand{\Th}[1]{Th.~\ref{thm:#1}}

\newcommand{\Sec}[1]{\S \ref{sec:#1}}
\newcommand{\Fig}[1]{Fig.~\ref{fig:#1}}
\newcommand{\Tbl}[1]{Table~\ref{tbl:#1}}
\newcommand{\App}[1]{App.~\ref{app:#1}}

\newcommand{\InsertFig}[4]
{\begin{figure}[h!t]
    \centerline{
     \includegraphics[width=#4]{./figures/#1}
    }
    \caption{{\footnotesize #2}
    \label{fig:#3}}
\end{figure}}

\newcommand{\InsertFigTwo}[5] {
\begin{figure}[h!t]
    \centerline{
     \includegraphics[width=#5]{./figures/#1}
     \hskip 0.5in
     \includegraphics[width=#5]{./figures/#2}
    }
    \caption{{\footnotesize #3}
    \label{fig:#4}}
\end{figure}}

\newcommand{\bR}{{\mathbb{ R}}}

\newcommand{\bT}{{\mathbb{ T}}}
\newcommand{\bZ}{{\mathbb{ Z}}}

\newcommand{\bS}{{\mathbb{ S}}}
\newcommand{\cA}{{\cal A}}
\newcommand{\cC}{{\cal C}}

\newcommand{\cO}{{\cal O}}
\newcommand{\cR}{{\cal R}}

\newcommand{\cU}{{\cal U}}
\newcommand{\cV}{{\cal V}}

\newcommand{\eps}{\varepsilon}

\newcommand{\Fix}[1]{\mathop{\rm Fix}\nolimits({#1})}

\newcommand{\TwoVec}[2]{{\begin{pmatrix}{#1} \\ {#2}\end{pmatrix}}}

\newtheorem{thm}{Theorem}
\newtheorem{lem}[thm]{Lemma}

\newcommand{\beq}[1]{\begin{equation}\label{eq:#1}}
\newcommand{\eeq}{\end{equation}}

\newenvironment{se}[1]{\equation\label{eq:#1}\aligned}{\endaligned\endequation}
\newcommand{\bsplit}[1]{\begin{se}{#1}}
\newcommand{\esplit}{\end{se}}

\newcommand{\thd}[1] {\multicolumn{1}{c}{#1}}
\newcommand{\thdBar}[1]{\multicolumn{1}{c|}{#1}}

\newenvironment{example}[1][]
 {
	\setlength \leftmargini {1.0em}		%controls the indentation for quote & list
	\setlength \topsep {0.5em}			%controls the top space for quote & list
	\begin{quote}
	{\it Example#1} }
	{\end{quote}
 }
\newcommand{\bexam}[1][:]{\begin{example}[#1]}
\newcommand{\eexam}{\end{example}}

\title{Mixed Dynamics in a Parabolic Standard Map}
\author{
 \begin{tabular}{cc}
 	L.~M.~Lerman\thanks
 	{
	LML was supported in part by RFBR (grant 14-01-00344a), the Russian Ministry of Science
	and Education (project 1.1410.2014/K, target part) and the Russian Science Foundation
  (project 14-41-00044). Useful conversations with V. Grines are gratefully acknowledged. Both
	authors acknowledge a support from CRDF (grant FSAX-14-60273-0).
  }
		 & J.~D.~Meiss\thanks
   {
    JDM was supported in part by NSF grant DMS-1211350.
    The suggestion of Robert Easton for the construction in
    \App{UnBounded} is gratefully acknowledged.
   }
 \\
Faculty of Mathematics \& Mechanics 	
		&	Department of Applied Mathematics\\
  Lobachevsky State University of Nizhny Novgorod,
  		&	University of Colorado \\
	Nizhny Novgorod, 603950 Russia
		&	Boulder, CO 80309-0526 \\
	lermanl@mm.unn.ru
		&	James.Meiss@colorado.edu\\
\end{tabular}  \\
}
\date{\today}
\begin{document}
\maketitle

\begin{abstract}
% \noindent
%\ams{34C37, 37C29, 37J45, 70H09}
%\pacs{02.40.-k, 05.45.-a, 45.20.Jj}
\vspace*{1ex}
\noindent
% Keywords:

We use numerical and analytical tools to demonstrate arguments in favor of the
existence of a family of smooth, symplectic diffeomorphisms of the two-dimensional torus
that have both a positive measure set with positive Lyapunov exponent and a positive measure
set with zero Lyapunov exponent. The family we study is the unfolding of an almost-hyperbolic
diffeomorphism on the boundary of the set of Anosov diffeomorphisms, proposed by Lewowicz.
\end{abstract}

%%%%%%%%%%%%%%%
%% Introduction
%%%%%%%%%%%%%%%
\section{Introduction}\label{sec:Introduction}
 Due to an extremely complicated intermixture of regular and chaotic orbits, the problem of
the orbit structure of a generic, smooth symplectic map remains mainly open, even for the two-dimensional case.
When the map is sufficiently smooth, its phase space typically exhibits both
regular dynamics due to invariant KAM curves (for instance, in the neighborhood of elliptic
periodic orbits) and seas of chaotic orbits (which numerical investigations indicate can be
densely covered by a single orbit). Moreover, such structures are observed---again in
numerical simulations---to occur at all scales. All this is well known and shown in many 
papers, for a
review see, e.g., \cite{Meiss92}. It is generally agreed that no tools currently exist that
allow one to rigorously elucidate the main points of this observed picture \cite{Sinai95}. Of
course, selected parts of this landscape can be explained; for example, KAM theory provides a
proof of the existence of invariant curves near generic elliptic periodic points. However
even for this case, there is essentially no rigorous characterization of the orbit behavior in the so-called chaotic zones---as depicted in Arnold's famous sketch \cite{Arnold63, Berry78}.

There has been much study of the destruction of invariant curves, and the resulting transition
from regular (quasiperiodic) to irregular (chaotic) behavior, in parameterized families of
area-preserving maps. Since a smooth invariant curve is not isolated, its destruction is
caused
by a loss of smoothness and, at least for twist maps, to the formation of a new, quasiperiodic
invariant Cantor set: an Aubry-Mather set \cite{Aubry83b,Mather82}. In many
families, one observes the ultimate destruction of all the invariant circles (of a given
homotopy class), and this leads to the study of the ``last" invariant curve, and the
development of Greene's residue criterion and renormalization theory \cite{MacKay93}.

At the opposite extreme, the ergodicity and hyperbolicity properties of Anosov diffeomorphisms
are well-understood \cite{Franks70}. This extreme of uniformly hyperbolicity can be thought
of as a complementary limit to integrability: the study of perturbations from ``anti-
integrability" was initiated in \cite{Aubry90}. Aubry's results are based on the consideration
of infinitely-degenerate diffeomorphisms and provide proofs of the existence of horseshoes;
however, they do not lead to proofs of a positive measure of chaotic orbits.

There have been attempts to understand the dynamics of symplectic diffeomorphisms on the torus
beyond the boundary of the Anosov maps \cite{Przytycki82}. Przytycki proved the existence of a
curve of diffeomorphisms that cross the Anosov boundary such that, outside the boundary, there
is a domain on the torus bounded by a heteroclinic cycle formed by merged separatrices of two
saddles that contains a generic elliptic fixed point. The remaining set of positive measure
has a nonhyperbolic structure and positive Lyapunov exponent. The drawback of this example is in
its infinite codimension in the space of smooth symplectic diffeomorphisms with $C^5$-topology:
the merging of separatrices of saddles is a codimension-infinity phenomenon. Przytycki's family
unfolds a smooth, almost-hyperbolic symplectic diffeomorphism of the torus proposed earlier by
Lewowicz \cite{Lewowicz80}. This diffeomorphism is a K-system that has positive Lyapunov
exponent \cite{Enrich01}.

This same trick (with the same drawbacks) was used later in \cite{Liverani04} to
construct symplectic diffeomorphisms arbitrarily close to Lewowicz's almost-hyperbolic map.
Smooth symplectic, transitive diffeomorphisms that are K-systems on closed two-dimensional manifolds
other than tori were constructed in \cite{Katok79} (see also, \cite{Gerber82}). Again it is
not clear how these results can be used to understand the orbit structure for a generic
diffeomorphism.

Following \cite{Lerman10}, we study the map
$f: \bT^2 \to \bT^2$, where $\bT = \bR / \bZ$, defined through
\beq{StdMap}
	f(x,y) = (x+y+g(x), y+g(x)) \mod 1.
%	x' &= x+ y' \\
%	y' &= y+g(x) .
\eeq
If the ``force" $g$ were a degree-zero circle map, then \Eq{StdMap} would be a generalized
Chirikov standard map \cite{Meiss92}. Instead, we assume that $g$ is a degree-one, circle map:
\beq{CircleMap}
	g(x+1) = g(x) + 1 .
\eeq
When $g$ is a monotone increasing diffeomorphism, \Eq{StdMap} is Anosov: every orbit is
uniformly hyperbolic and $f$ is topologically conjugate
to Arnold's cat map $a: \bT^2 \to \bT^2$,
\beq{Anosov}
	a(x,y)= A \begin{pmatrix} x \\ y \end{pmatrix} \mod 1,\quad \mbox{where}\quad
 	A=\begin{pmatrix}2&1\\1&1 \end{pmatrix} .
\eeq
More generally, Franks showed that \Eq{StdMap} with \Eq{CircleMap} is semi-conjugate to $a$
\cite{Franks70}, i.e., there is a continuous,
onto map $k:\bT^2 \to \bT^2$ such that
\beq{SemiConjugacy}
	k \circ f = a\circ k.
\eeq
The map $k$ depends continuously on $g$, and when $g$ is strictly monotone, $k$ is a
homeomorphism, implying---as mentioned above---that $f$ is then conjugate $a$.

In \cite{Lerman10}, the first author made an attempt to elucidate the features of
\Eq{StdMap} when the circle map $g$ acquires a critical fixed point,
\beq{Parabolic}
	g(x_p) = Dg(x_p) = 0 .
\eeq
In this case \Eq{StdMap} has a parabolic fixed point $p = (x_p,0)$ and is no longer Anosov.
The main
result of \cite{Lerman10} was to show that the diffeomorphism acquires elliptic behavior when
it
crosses this Anosov boundary. Another feature of this map is the separation of the phase space
into
two regions, one in which the dynamics is nonhyperbolic and the other in which the
diffeomorphism
appears to be nonuniformly hyperbolic. Though neither of these statements were proved in
\cite{Lerman10}, considerations in favor of these statements were presented.

In this paper we try to use numerical methods to substantiate the following assertions about
\Eq{StdMap} under the assumptions \Eq{CircleMap} and
\Eq{Parabolic}.
\begin{itemize}
	\item There is an invariant, open region $E \subset \bT^2$ whose boundary is
	formed from the stable and unstable manifolds of two fixed points of the map,
	a hyperbolic saddle, $h$, and a parabolic point, $p$. The Lebesgue measure of $E$
	is strictly less than that of $\bT^2$.
		
	\item The channel $E$ contains all non-hyperbolic orbits of $f$, and indeed has
	elliptic orbits for generic $\eps>0$.
	
	\item Conversely, the dynamics of $f|_H$, where $H=\bT^2\setminus E$, is
  nonuniformly hyperbolic; that is, the map is ergodic in $H$ and has positive
	Lyapunov exponent.

\end{itemize}
Of course, these statements are purely numerical observations, which should therefore be
considered mathematically as conjectures.

\section{A Parabolic Standard Map}\label{sec:Parabolic}

Following \cite{Lewowicz80, Lerman10}, we study the dynamics of \Eq{StdMap} using the
degree-one circle map
\beq{Force}
	g(x) = x +\frac{1}{2\pi}\left[\mu - (1+\eps) \sin(2\pi x)\right] ,
\eeq
where $\mu$ and $\eps \ge 0$, see \Fig{gplot}. Note that when $\eps = -1$ and $\mu = 0$, the
map \Eq{StdMap} reduces
to Arnold's cat map \Eq{Anosov}.

%%%%%
\InsertFig{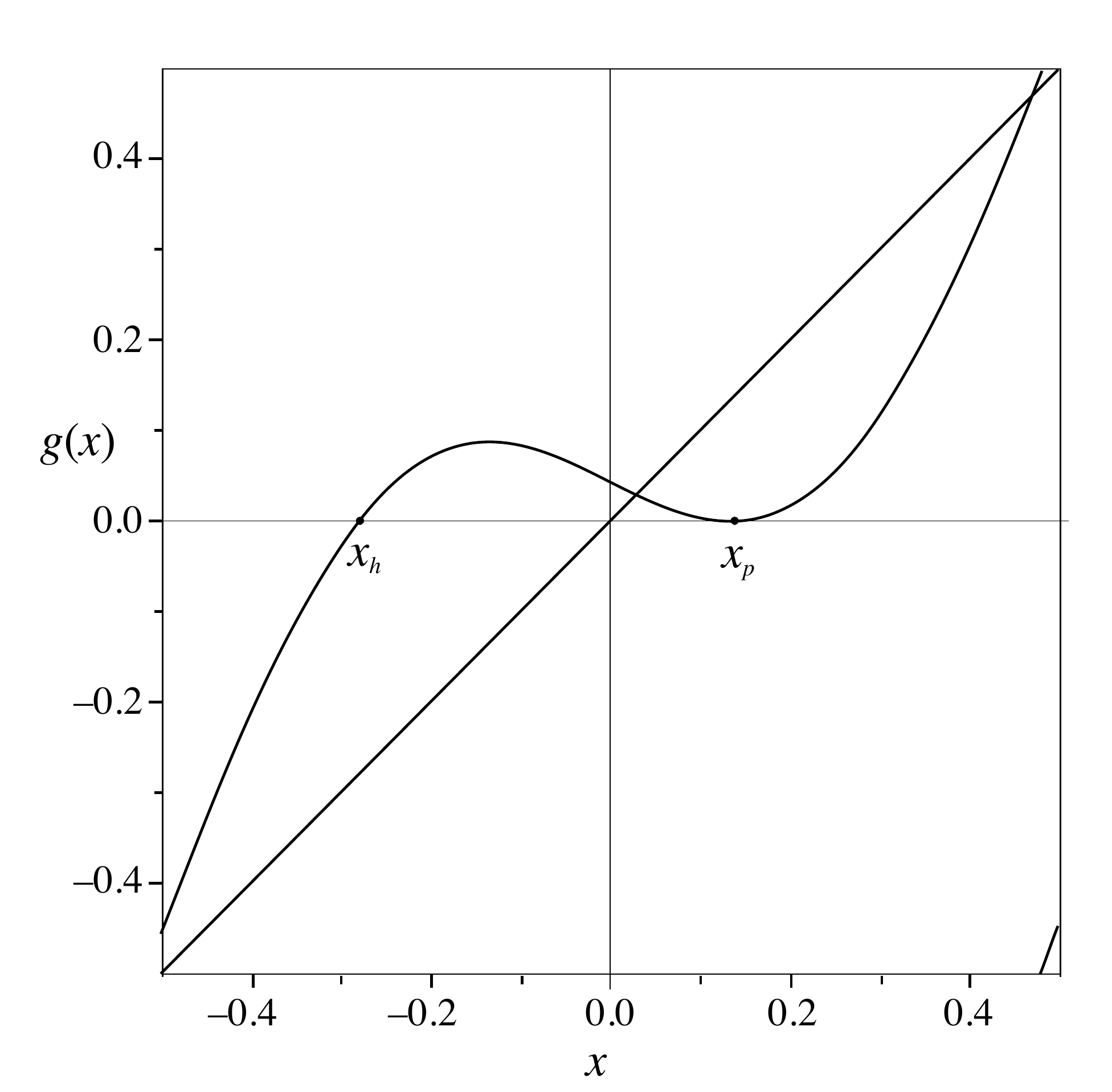}{Graph of the force $g$ \Eq{Force}, for $\eps = 0.5$ and $\mu=\mu_p(\eps)
\approx 0.27696$ from \Eq{Mu_p}, with the parabolic
point $x_p \approx 0.13386$ and saddle $x_h \approx -0.27889$.}{gplot}{2.5in}
%%%%%

The map $f$ is a diffeomorphism whenever $g$ is smooth. Indeed
\[
	f^{-1}(x,y) = (x-y, y-g(x-y)) .
\]
Moreover, this map is reversible, $f \circ S = S \circ f^{-1}$, with the ``second" reversor of
Chirikov's map (It does not have the first reversor since $g$ is not odd when $\mu \neq 0$.),
\beq{Reversor}
	S(x,y) = (x-y,-y) ,
\eeq
with the fixed set $\Fix{S} = \{(s,0): s \in \bS\}$. Note that since $S$ is
an involution, the map
\[
	f \circ S(x,y)= S\circ f^{-1}(x,y)= (x-2y+g(x-y), -y+g(x-y))
\]
is also a reversor, with the fixed set
\[
	\Fix{f\circ S} = \{s +\tfrac12g(s),\tfrac12g(s): s \in \bS\} .
\]

Under the assumption \Eq{Parabolic}, $g(x-x_p) = \cO((x-x_p)^2)$, the map \Eq{StdMap} has a
(symmetric) parabolic fixed point $p = (x_p,0)$.
For the case \Eq{Force} this fixed point occurs at
\beq{ParabolicFixedPt}
	x_p = \frac{1}{2\pi}\sec^{-1}(1+\eps)
	  = \frac{1}{\pi}\sqrt{\frac{\eps}{2}}\left(1 - \frac{5}{12}\eps + \cO(\eps^2)\right),
\\
\eeq
when $\mu$ is chosen to be
\bsplit{Mu_p}
	\mu = \mu_p(\eps) &\equiv (1+\eps)\sin(2\pi x_p) - 2\pi x_p
	  = \sqrt{\eps(2+\eps)}- \sec^{-1}(1+\eps) \\
	  &= \frac{\sqrt{8}}{3}\eps^{3/2} \left(1-\frac{9}{20}\eps + \cO(\eps^2)\right) .
\esplit
Note that since $Dg(x_p) = 0$, the Jacobian
\beq{Jacobian}
	Df = \begin{pmatrix} 1+Dg(x) & 1 \\ Dg(x) & 1 \end{pmatrix}
\eeq
at $x_p$ has a double eigenvalue $1$ with a nontrivial Jordan block. Moreover, as was shown in
\cite{Lerman10}, the first nonzero coefficient
in the nonlinear normal form is quadratic in $x-x_p$.

We will fix $\mu = \mu_p(\eps)$ using \Eq{Mu_p} and think of $f$ as a one-parameter family $f_
\eps$. This family also has a (symmetric) hyperbolic saddle fixed point $h = (x_h,0)$ where
$x_h < 0$ is the negative root of $g(x)$,
\beq{SaddleFixedPt}
	x_h = -\frac{\sqrt{2\eps}}{\pi}\left(1 - \frac{19}{60}\eps + \cO(\eps^2)\right).
\eeq
This is a hyperbolic fixed point of $f$ since $Dg(x_h) > 0$ whenever $\eps >0$; indeed its
multipliers are
\[
	\lambda_{\pm} = 1 \pm \sqrt{3\eps} +\frac{3}{2}\eps + \cO(\eps^{3/2}),
\]
The pair of fixed points $(x_p,0)$ and $(x_h,0)$ are born from a degenerate saddle at the
origin for $\eps = \mu = 0$.

Additional fixed points may exist when $\eps$ is large enough so that $g(x)$ is a nonzero
integer for some $x \in [-\tfrac12, \tfrac12]$.
For example a saddle-center bifurcation creates a pair of fixed points near $(x,y)
=(-0.2151,0)$ when $\eps\approx 3.603$. However,
we will study the case that $\eps$ is much smaller than this so that $f$ has only two fixed
points.

%%%%%%%%%%%%%%%%%
%% Invariant Manifolds
%%%%%%%%%%%%%%%%%
\section{Invariant Manifolds and the Channel}\label{sec:InvariantManifolds}

The region $E$ mentioned in \Sec{Introduction} is an open domain that is bounded by segments
of the stable and unstable manifolds of the parabolic and saddle fixed points $p$ and $h$. In
a sense, the construction we follow is similar that of a ``DA-map" on the torus
\cite{Smale67}. In that case, the unstable manifold of a saddle of an Anosov map is blown-up
to make a \emph{channel} $E$ in such a way that it contains the original unstable fixed point,
and its boundaries become unstable manifolds of two saddles. However, since \Eq{StdMap} is
symplectic we cannot blow up a single manifold; instead, the curve to be blown-up corresponds
to the right going halves of the unstable $W^u(O)$ and stable $W^s(O)$ manifolds of the
degenerate saddle $O = (0,0)$ of
the almost-hyperbolic diffeomorphism $f_0$ at $\eps = \mu = 0$. Near the degenerate saddle,
these two right-going curves form a cusp-shaped separatrix. Since $f_0$ is conjugate to an
Anosov map, these manifolds are dense on the torus, and they intersect on a dense set of
homoclinic points. Moreover, at every intersection point (except for the degenerate saddle
itself), the manifolds cross transversely \cite{Lewowicz80}. As a consequence, if we thicken 
these two curves we
must get a Cantor-like set on the torus. Our simulations support these conclusions, though the
situation is more complex. To avoid these complications, we do not attempt this blow-up
construction, but instead use a bifurcation technique.

As noted in \Sec{Parabolic}, when $\eps >0$ the family $f_\eps$ has two fixed points: a
hyperbolic saddle $h$ and a parabolic point $p$. Each fixed point has smooth, stable,
$W^s$, and unstable, $W^u$, manifolds; numerical approximations are shown in
\Fig{ParabolicManifolds}.
The existence of the manifolds of the saddle follows from the standard stable manifold theorem
\cite{Hirsch77}.
For the parabolic point, existence follows from the results of Fontich \cite{Fontich99}. In
contrast to the saddle, the curves $W^{s,u}(p)$ are one-sided and they start as a cusp of the
form
\[
	W^{s,u}(p) = \{ (x, \alpha (x-x_p)^{3/2}) + \cO(\eps): x \ge x_p\}
\]
with $\alpha = \mp \sqrt{\frac{2\pi}{3}} [\eps(2+\eps)]^{1/4}$, see \App{ParabolicManifolds}.
As will be discussed in \Sec{Slope}, the manifolds of the parabolic point are asymptotic to
the right-going branches of the manifolds of the hyperbolic point; this is seen numerically
when the manifolds are extended, as shown in \Fig{NumericalMan}.

%%%%%
\InsertFig{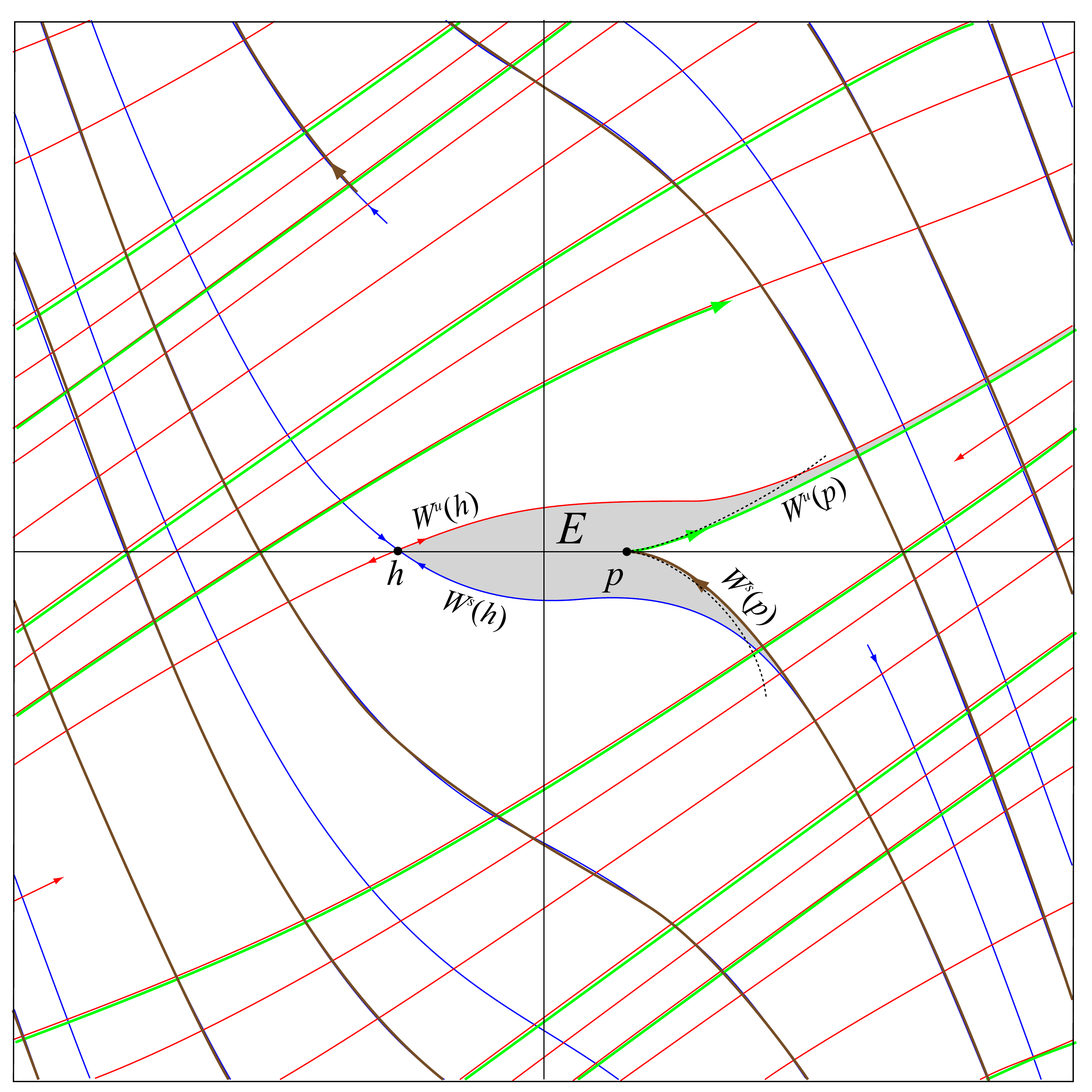}{Stable (brown) and unstable manifolds (green) for the parabolic
fixed point $p =(x_p,0)$ and
for the saddle $h = (x_h,0)$ (blue and red, respectively) for the map \Eq{StdMap} with
\Eq{Force}, $\eps = 0.1$, and
$\mu = \mu_p(\eps) \approx 0.029558$. Also shown (dashed curves) is a $10^{th}$ degree
polynomial approximation \Eq{ParabolicApprox}
for the parabolic manifolds. The gray region is the channel $E$.}{ParabolicManifolds}{3in}
%%%%%

Since the right-going unstable manifolds of $p$ and $h$ are born out of the degenerate
manifold $W^u_0$, one can prove that they depend continuously on $\eps$;
moreover, any finite segment of these manifolds detached from the fixed point can be proven to
vary continuously in the $C^1$-topology.
In particular, if $T$ is any segment transverse to the
local unstable manifold of the degenerate saddle, then the local manifolds $W^u(h)$ and
$W^u(p)$ that emerge will continue to transversely intersect $T$ when $\eps$ is small enough.

Since the map has the reversor $S$ \Eq{Reversor}, and the points $h$ and $p$ are symmetric,
their corresponding stable manifolds are obtained by applying $S$ to the unstable manifolds,
and the same continuity and transversality conditions apply.

%%%%%
\InsertFigTwo{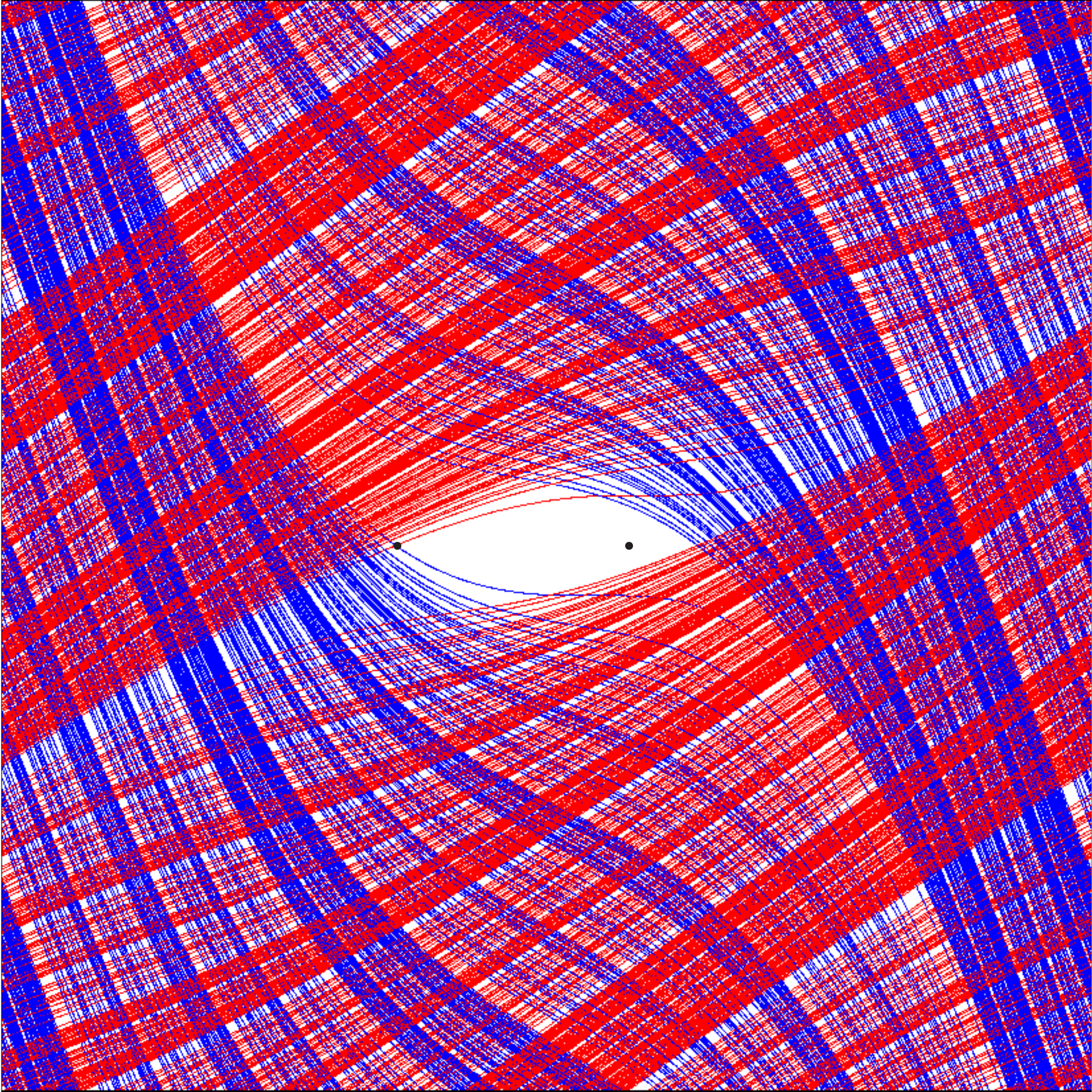}{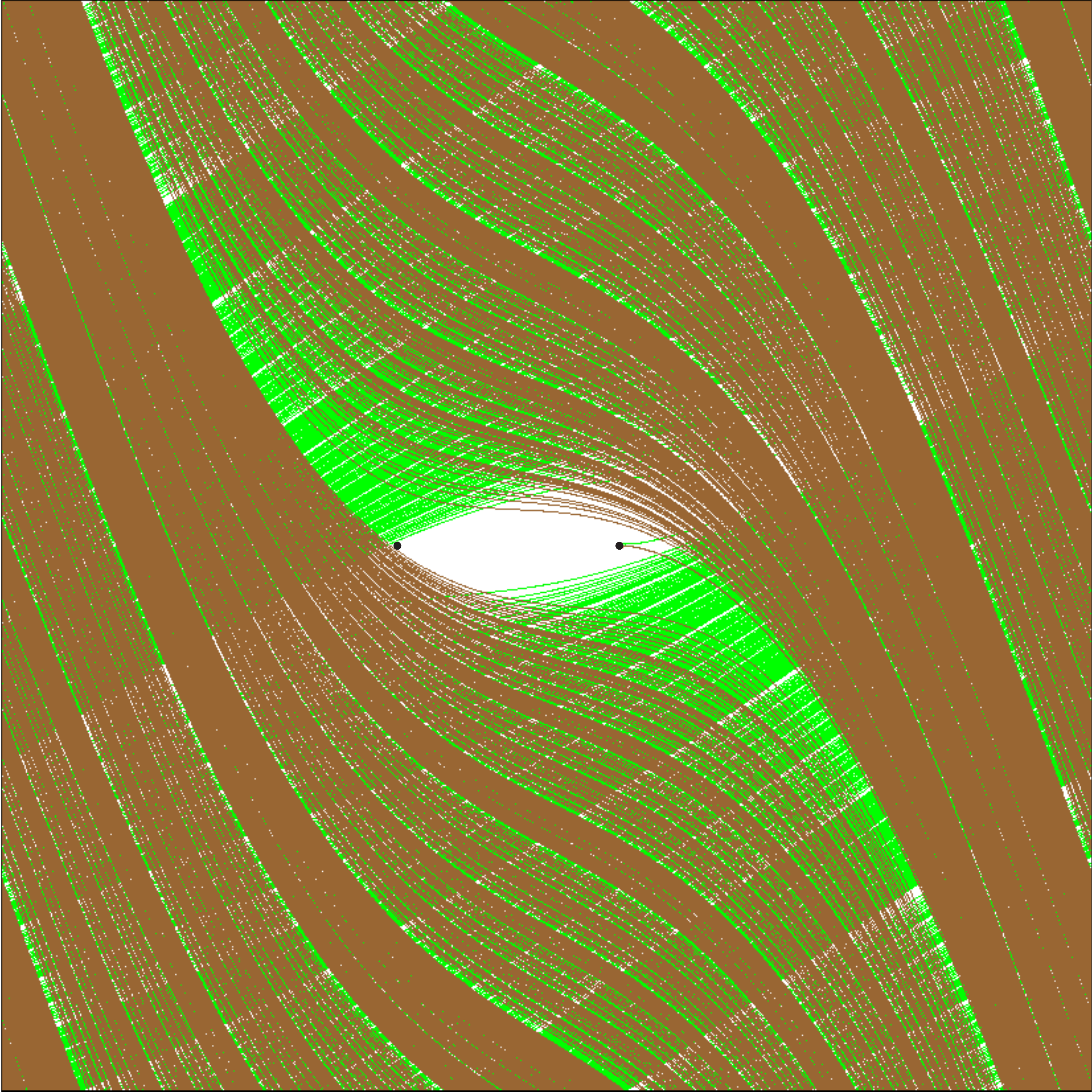}{Numerically computed stable and unstable manifolds for
the saddle $(x_h,0)$ (left) and
the parabolic point $(x_p,0)$ (right) for the map \Eq{StdMap} with \Eq{Force}, $\eps = 0.1$,
and $\mu = \mu_p(\eps)$.}{NumericalMan}{3in}
%%%%%

The boundary of the initial part of the upper half, $E^+$, 
of the channel $E$ can be constructed by
following the right-going manifold $W^u(h)$ from $h$ until it first intersects the transverse
segment $T = \{(\tfrac12, y): y\in [0,\tfrac12)\}$. For small enough $\eps$, this segment is
transverse $W^u(p)$ as well, so follow $T$ to $T \cap W^u(p)$ and then continue along the
parabolic manifold back to $p$. Connecting $p$ to $h$ along the $x$-axis, closes the boundary
of the initial part of $E^+$. A similar construction for the stable right-going manifolds
using $T' = \{(x,-\tfrac12): x\in [0,\tfrac12]\}$, gives the lower half, $E^-$. The union is
shown in gray in \Fig{ParabolicManifolds}. The remainder of the channel is then obtained by
forward and backward iteration.

One has to remark that the behavior of manifolds for $h$ and $p$ is considerably more complex
than that for $\eps=0$---especially because of the formation of folds. This is more easily seen for larger
values of $\eps$: \Fig{FourManifoldsE05} shows the manifolds for $\eps= 0.5$. The size and
number of these folds on any finite segments grows with $\eps$.

%%%%%
\InsertFig{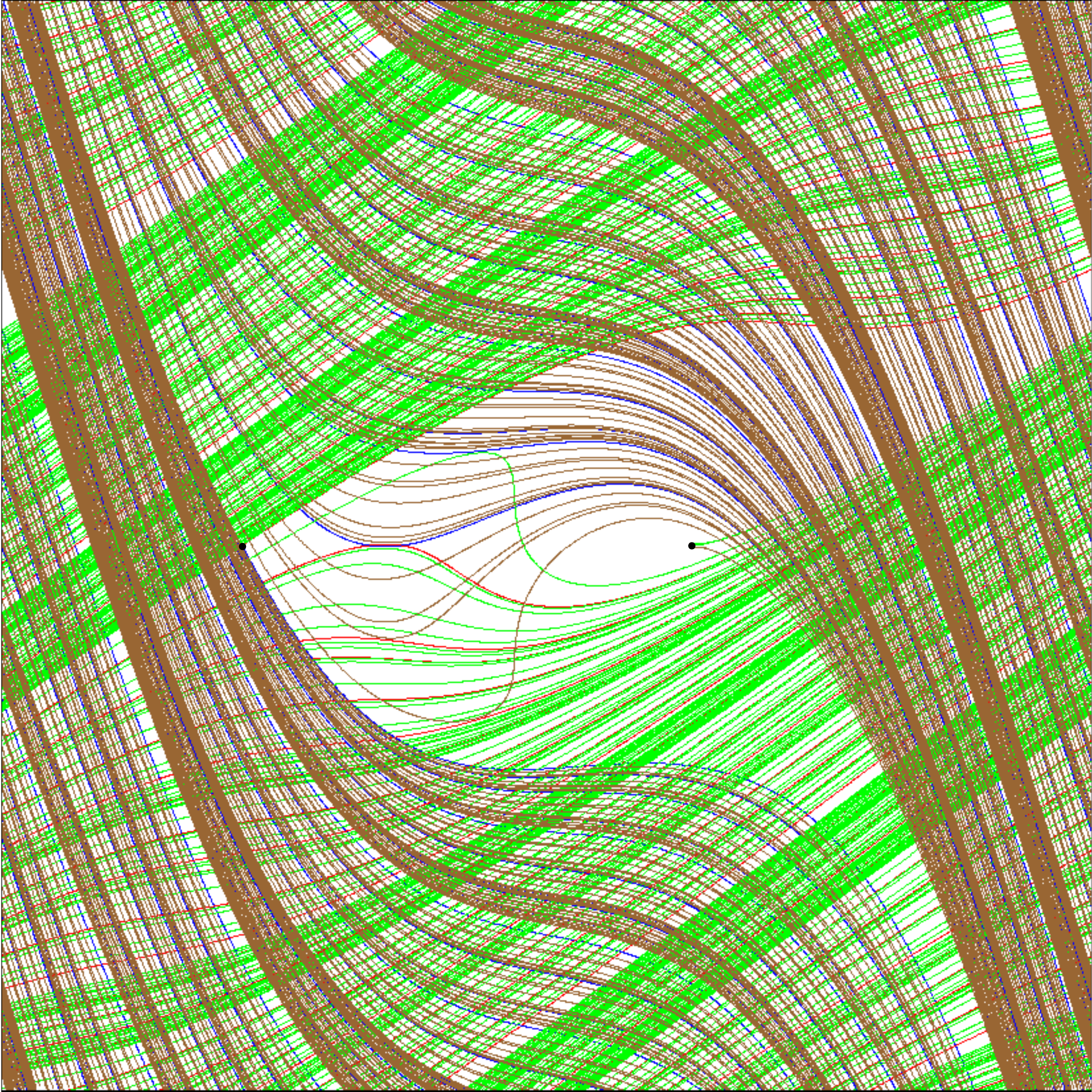}{Numerically computed stable and unstable manifolds for the saddle
(blue and red) and parabolic (brown and green)
fixed points for \Eq{StdMap} with \Eq{Force}, $\eps = 0.5$, and $\mu = \mu_p(\eps)$.
The manifolds for the saddle $h$ are mostly obscured because they were plotted before those of
$p$.}{FourManifoldsE05}{3in}
%%%%%

These folds are caused by homoclinic intersections. To see this, choose a segment of $W^u(O)$
for the almost hyperbolic diffeomorphism $f_0$ that is long enough to transversely intersect
$W^s(O)$ at a primary homoclinic point $q_0$.
Since the perturbed manifolds of $p$ and $h$ are $C^1$-close to the degenerate manifolds when
$\eps >0$, the unstable manifolds will transversely intersect both of the stable channel
boundaries near $q_0$.
This gives rise to a rectangle $R$, a portion of the channel that is
re-injected into itself, see \Fig{ChannelFolds}. Two of the corners of $R$ are homoclinic to
$p$ and $h$, and two are heteroclinic from/to $p$ and $h$.
Since iteration causes points to move monotonically along the stable manifolds, the corners on
$W^s(h)$ monotonically tend to $h$, and the two on $W^s(p)$ monotonically approach $p$.
However, the interiors of the sides of the rectangle formed from the unstable manifolds are
inside $E$, so they must move along the channel and tend to $W^u(h)$. Therefore, these
segments lengthen along $W^u(h)$, but have endpoints tied to $W^s(p)$. This makes them form a
fold, like that in \Fig{ChannelFolds}.

For each such fold in the unstable manifolds, there is a corresponding fold in the stable
manifolds obtained by applying the reversor \Eq{Reversor}. Suppose two such symmetric folds
intersect. If one decreases $\eps \to +0$, these two symmetric folds tend to the related
points on the unstable and stable separatrices of the degenerate saddle. A consequence is that
these two pieces of the stable and unstable manifolds of the saddle have to touch each other
for some value of $\eps>0$ on some point of the symmetry line $y=0$. This gives a homoclinic
tangency. Near tangencies can be seen in \Fig{FourManifoldsE05}. This story requires a proof,
of course, but shows one mechanism for the creation of a tangency.

%%%%%
\InsertFig{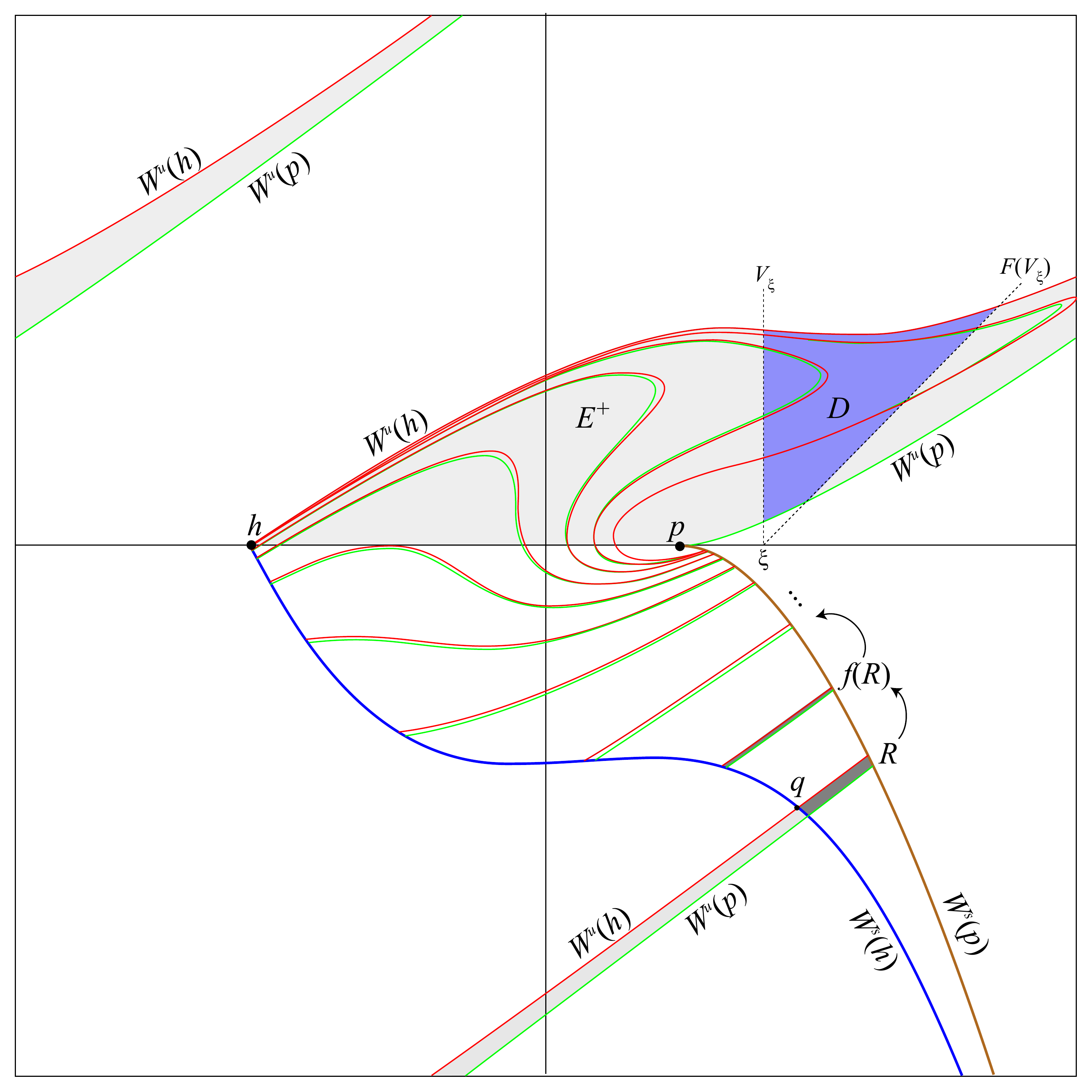}{Channel for $\eps=0.5$ and $\mu_p(0.5)$. Shown are segments of the
stable and unstable manifolds of $h$ and $p$ as in \Fig{ParabolicManifolds}. The unstable
channel, $E^+$ is shaded gray, and a fundamental domain, $D$, is blue. The curves show the
first re-injection of the unstable manifolds into the channel and subsequent
formation of a fold. The
unstable manifolds are clipped when they first cross the stable boundary of the channel $E$.
The transversal crossing creates a rectangle $R$ (dark gray) whose corner points are
heteroclinic from/to $h$ or $p$. A primary homoclinic point $q$ is labeled. Nine images of $R$ are shown. Upon iteration the unstable sides of the rectangle begin to fold due to their accumulation onto $W^u(h)$.}{ChannelFolds}{3.5in}
%%%%%

%%%%%%%%%%%%%%%%%
%% Elliptic Dynamics
%%%%%%%%%%%%%%%%%
\section{Elliptic Dynamics}\label{sec:EllipticDynamics}

It is known that the quadratic homoclinic tangencies (as can be seen in 
\Fig{FourManifoldsE05}) and the resulting Newhouse phenomena \cite{NewHouse77, Duarte08, 
GTSh07} create elliptic orbits. However, there are no visible elliptic orbits in the 
numerical computations when $\eps$ is small. As is seen in \Fig{TwoManifoldsE0379}, when $\eps
$ is not too large the map appears numerically to be ergodic in the sense that a single orbit 
lands in every pixel in a computer generated image. The first visible loss of ergodicity on 
this scale is due to a period-five, saddle-center bifurcation at $\eps \approx 1.159$, which 
creates a chain of five elliptic islands inside the channel. These islands undergo the usual 
sequence of area-preserving resonant bifurcations, before being destroyed by period-doubling 
near $\eps = 1.274$.

%%%%%
\InsertFig{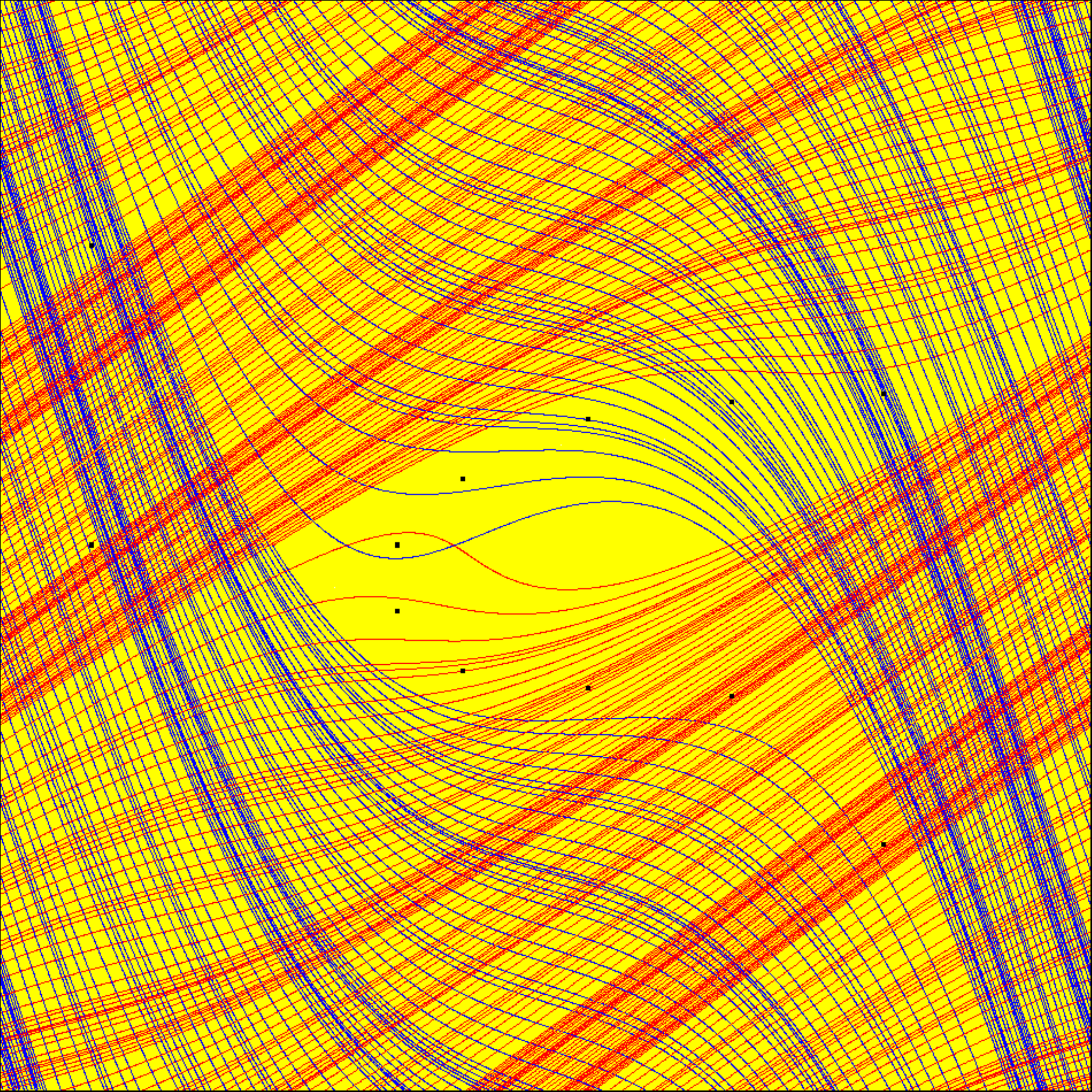}{Numerically computed stable (blue) and unstable (red) manifolds the hyperbolic point
$h$ of \Eq{StdMap} with \Eq{Force}, $\eps = 0.37970035$, and $\mu = \mu_p(\eps)$. The background (yellow) shows
$10^7$ iterates of $(0,0)$; these cover every pixel in the $950\times950$ pixel image. Nevertheless, there is
an elliptic period-$12$ orbit in the channel; it is covered by the black boxes (the boxes are much larger than
the islands!). An enlargement of one of these islands is in \Fig{Period12}.}
{TwoManifoldsE0379}{3in}
%%%%%

For smaller $\eps$, the islands are too small---in size and/or interval of stability---to be easily observed numerically.
Indeed, for $0.01 \le \eps \le 1.15$ in steps of $\Delta \eps = 0.01$, we observe no islands in the phase space whose size is larger than $10^{-4}$.

%To show the existence of elliptic dynamics inside the channel \cite{Lerman10} used a
%construction that allowed one to find the homoclinic tangency
%for a period-two saddle orbit. In fact the same can be done through the tangency of the fixed
%point $h$. Such a tangency is clearly seen in
%\Fig{ChannelFolds} for the fifth image of the region $R$---since $W^u(h)$ is tangent to the
%fixed set $y=0$ of $S$, so is $W^s(h)$. As is well-known, such a homoclinic tangency is
%generically accompanied by the existence of elliptic periodic orbits for parameters intervals
%close to that of a homoclinic tangency \cite{Mora97, Duarte08}. The resulting elliptic orbits
%would lie inside the channel.

Nevertheless, such orbits can be found using the reversibility of the map
and thus exploiting its geometric properties. To this end,
recall an assertion from the theory of reversible maps \cite{Lamb98}:

%%%%%%
\begin{thm}[Devaney, 1976]
Suppose that $\cO(z)$ is a symmetric periodic orbit of a reversible diffeomorphism
$f$ with a reversing involution $S$ having a smooth submanifold of fixed points $\Fix{S}$.
Then:
\begin{itemize}
\item $\cO(z)$ has period $2p$ if and only if there
exists $\zeta \in \cO(z)$ such that $\zeta \in \Fix{S}\cap f^p(\Fix{S})$ or $\zeta \in \Fix{f
\circ S}\cap f^p(\Fix{f\circ S});$

\item $\cO(z)$ has period $2p+1$ if and only if there
exists $\zeta \in \cO(z)$ such that $\zeta \in \Fix{S}\cap f^p(\Fix{f\circ S})$.
\end{itemize}
\end{thm}
%%%%%%%

Consequently, in order to find a symmetric periodic orbit one should search among the
intersections of the images of $\Fix{S}$ and $\Fix{f \circ S}$. In the
two-dimensional case, if the intersection of these two sets is transverse then this orbit is a
saddle or an elliptic point. If this intersection point is a quadratic tangency, then the
orbit is parabolic, and generically corresponds to a saddle-center bifurcation upon varying a
parameter. The proof of this statement is in the \App{creation_elliptic}.

This indeed is confirmed by calculation.\footnote
%%%%%
{the authors thank A. Kazakov for his supporting numerics made upon our request.}
%%%%%
For example, using the reversor \Eq{Reversor}, we found a pair of
period-$12$ orbits born on $ \Fix{S} \cap f^6(\Fix{S})$ in a saddle-center bifurcation at $
\eps \approx 0.37970034$. For slightly larger $\eps$, one of these orbits is elliptic, and is
surrounded by an island, see \Fig{Period12}. As is typical, this island is also surrounded by
other elliptic points---a prominent period-$12 \times 7$ orbit is visible---and this orbit
also surrounded by more elliptic periodic orbits, for example one of period $12 \times 7
\times 6 = 504$. This island has stable orbits only for a parameter window of width $\Delta
\eps \approx 2(10)^{-7}$.
%disappears near $\eps = 0.37970054$.

We conjecture that there are similar, higher-period, elliptic orbits for arbitrarily small,
positive $\eps$. However, even a numerical investigation of this seems very difficult.

%%%%%
\InsertFig{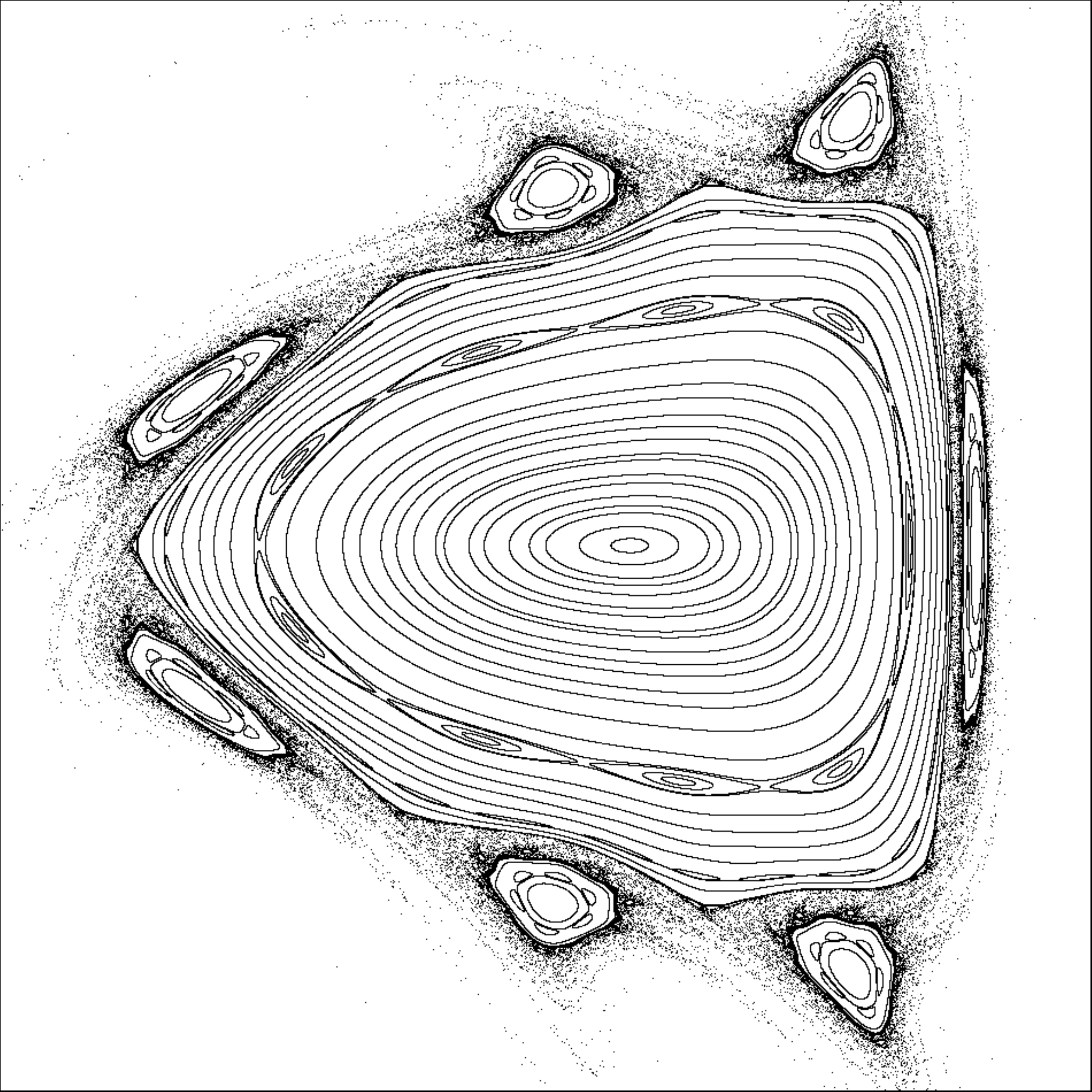}{Island around a period-$12$ elliptic point for \Eq{StdMap} with
\Eq{Force} for $\eps = 0.37970035$ and $\mu = \mu_p(\eps)$.
The bounds of this figure are $[-0.13527,-0.13505] \times [-5.5(10)^{-8}, 5.5(10)^{-8}]$.}
{Period12}{3in}
%%%%%

%%%%
%\InsertFigTwo{TwoManifoldsE12}{Period5E1161}{Numerically computed stable and unstable
%manifolds for the saddle of \Eq{StdMap} with \Eq{Force} and $\mu = \mu_p(\eps)$ and (a) $\eps
%= 1.2$, (b) $\eps = 1.161$. The background (yellow) shows the image of a single orbit, which
%covers every pixel images with the exception of those in a period-$5$ island chain, shown
%near its tripling bifurcation in (a). In the enlargement (b), the bounds of the plot are $
%[-0.3,0.2] \times [-0.05,0.05]$. Stable and unstable manifolds of the saddle show a near
%tangency.}
%{TwoManifoldsE12}{3in}
%$[0.315,0.432] \times [0.294,0.445]$ (for the 1.16 enlargement)
%%%%%%

%%%%%%%%%%%%%%%%%
%% Channel Slope
%%%%%%%%%%%%%%%%%
\section{Asymptotic Behavior of the Channel}\label{sec:Slope}
Let $F: \bR^2 \to \bR^2$, be a lift of $f$, chosen so that its fixed points, which we will
call $\tilde{h}$ and $\tilde{p}$, lie in the fundamental square $[-\tfrac12,\tfrac12) \times
[-\tfrac12,\tfrac12)$. The lifted unstable manifolds, $W^u(\tilde{h})$ and $W^u(\tilde{p})$,
see \Fig{Strips_eps06}, are observed to be asymptotic to a ray with fixed slope, and indeed
to be asymptotic to one another. Of course, reversibility implies that the same properties
hold for the stable manifolds (with another slope).

To explain this, we first parameterize the manifold $W^u(\tilde{h})$ in the following standard
way. Take some point $(x_h, y_h) = \zeta_h \in W^u(\tilde{h})$, and let
\[
	\ell=W^u(\zeta_h, F(\zeta_h))
\]
be the fundamental segment of unstable manifold connecting $\zeta_h$ to $F(\zeta_h)$.
By selecting $\zeta_h$ close enough to $\tilde{h}$ it is always possible to have $\ell$ belong to the unit
square. We parameterize $\ell = \{\zeta_h(s): s \in [0,1]\}$ so that $\zeta_h(0) = \zeta_h$ and
$\zeta_h(1) = F(\zeta_h)$. This parameterization can be extended to $\bR$ using iteration so
that $s+1$ corresponds to the image.
\[
	\zeta_h(s+1) = F(\zeta_h(s)),
\]
and $s-1$ to the preimage, etc. Since the iterates of $\ell$ cover $W^u(\tilde{h})$, we get a
full parameterization: \beq{parameterize}
	W^u(\tilde{h}) = \{\zeta_h(s) = (x_h(s),y_h(s)): s \in \bR, \}.
\eeq
Note that $\zeta_h(s) \to \tilde{h}$ as $s\to -\infty$.

Computing the manifolds up to $x = 1000$ and for values $0.01\le \eps \le 1.0$, we observe that
\beq{Slopes}
	\frac{y_{h}(s)}{x_h(s)} = \phi^{-1} + \cO\left(\frac{1}{x_h(s)}\right) ,
\eeq
where $\phi = \tfrac12(1+\sqrt{5})$ is the golden mean. The same considerations apply to the unstable manifold of the parabolic point
$W^u(\tilde{p})$. To verify these observations, we first recall that the map \Eq{StdMap} is
semi-conjugate to the linear, Anosov map $a$, \Eq{Anosov}, i.e., that there is a continuous,
onto map $k$ that is homotopic to the identity such that \Eq{SemiConjugacy} holds
\cite{Lerman10}. In particular, this implies that $f$ induces a map $f^* = a^*$ on the 
fundamental group $\bZ^2$ of the torus that is hyperbolic.
The Anosov map $a$ has a unique fixed point at the origin, a saddle with
eigenvalues $\phi^2$ and $\phi^{-2}$. The right-going unstable manifold of the origin, is the
projection of the right-going unstable eigenvector of the matrix $A$ onto the torus. This has
slope $\phi^{-1}$---and of course this slope is precisely the slope that we observe in
\Eq{Slopes}.

%%%%%
\InsertFig{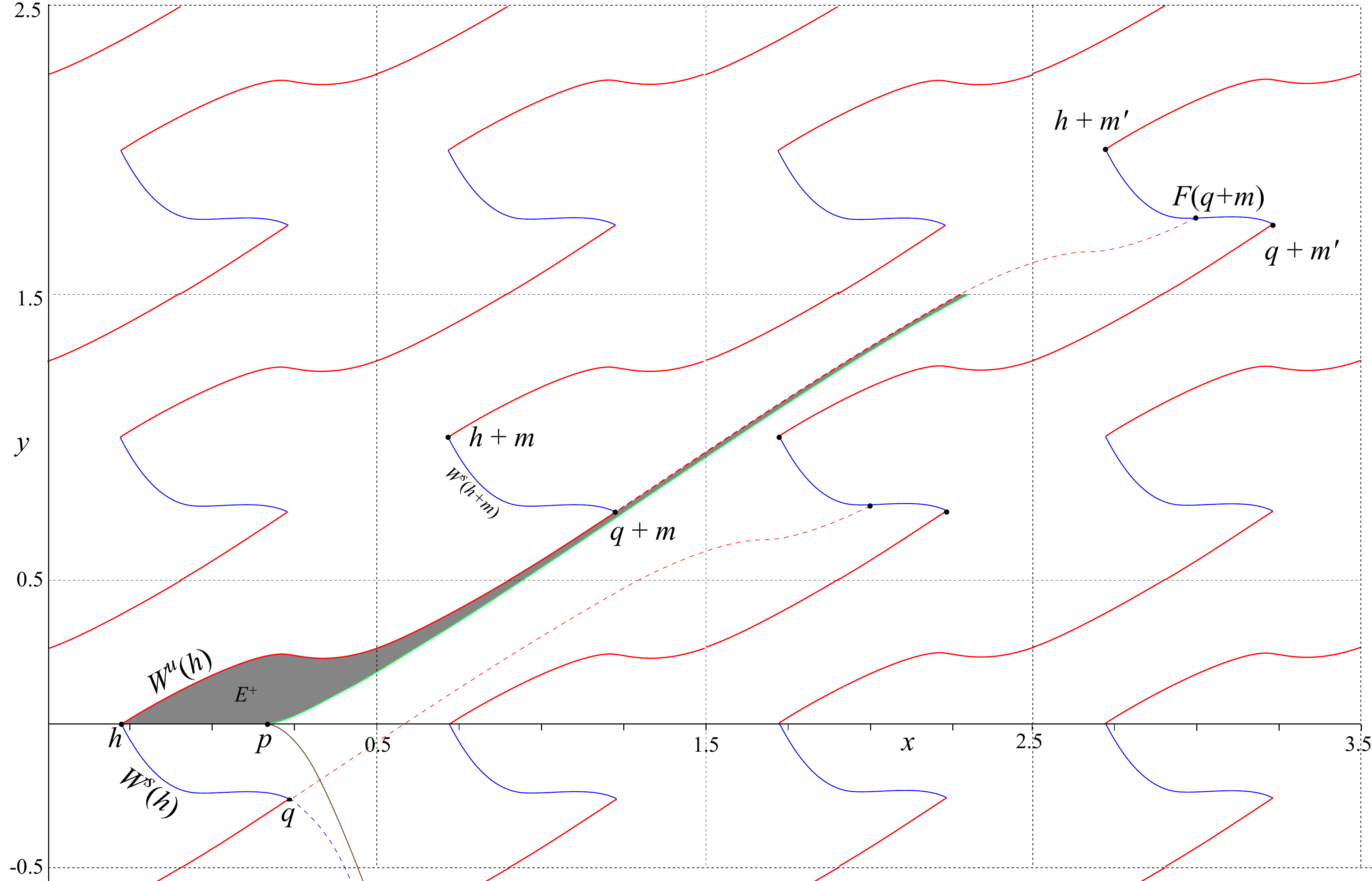}{Stable (blue) and unstable (red) manifolds of the hyperbolic fixed
point, and stable (brown) and unstable (green) manifolds of the parabolic fixed point for 
$\eps =0.6$ for the lift $F$. The gray region, labeled $E^+$, is a portion of the upper half 
of the lifted channel. Also shown are the strips constructed from the lifts of manifold 
segments $W^{u,s}(h,q)$ to a homoclinic point $q$. The image $F(q+m)$ lies on the stable 
segment $W^u(h+m',q+m')$ where $m' = (3,2)$.}{Strips_eps06}{6in}
%%%%%

We will prove the following result.
%%%%
\begin{thm}\label{thm:Slope}
Let $F$ be a lift of \Eq{StdMap} to the plane and assume that the force $g$ is of the
form \Eq{Force} with $\eps$ small enough. Then the right-going unstable
manifolds of the saddle $\tilde{h}= (x_h,y_h)$ and parabolic $\tilde{p} = (x_p,y_p)$ fixed
points of $F$ both tend to $\infty$ and have a limiting slope $\phi^{-1}$, the inverse of the
golden mean.
\end{thm}
%%%%
As a start to the proof of \Th{Slope}, we first note that the semi-conjugacy collapses the
stable and unstable manifolds of both $p$ and $h$ onto the corresponding manifolds of the
fixed point of the Anosov map:

\begin{lem}\label{lem:SemiManifolds}
The semi-conjugacy $k$ transforms $W^u(h)$, the upper boundary of
the channel $E$ for $f$, onto the unstable manifold $\gamma = W^u(O)$ of $O= (0,0)$
under the Anosov map $a$. The same is true for $W^u(p)$. Similarly both $W^s(h)$ and $W^s(p)$
are mapped by $k$ onto the stable manifold of $O$ under $a$.
\end{lem}

\begin{proof}
The map $f$ has exactly two fixed points: $h$, and $p$, and the map $a$ has the unique fixed 
point, $O$. Using \Eq{SemiConjugacy}, the $k$-image of a fixed point for $f$ must be a fixed 
point for $a$, and hence $k(p) = k(h) = O$. By definition, the backward orbit of any point $
\zeta \in W^u(h)$ tends to $h$: $ \zeta_t = f^{t}(\zeta)\to h$ as $t \to -\infty$. Under the 
semi-conjugacy, one has $k(\zeta_t) = a^t\circ k(\zeta)$, thus the $k$-image of the backward 
orbit of $\zeta$ is the backward orbit of the point $k(\zeta)$. Since $k$ is continuous, 
$k(\zeta_t)\to k(h) = O$ as $t\to - \infty$, thus $k(\zeta) \in \gamma$. Thus $k(W^u(h))$ is a 
subset of $\gamma$. To prove that this image is onto $\gamma$ we need to find a fundamental 
segment in $W^u(h)$ whose image is a fundamental segment in $\gamma$. One way to do this is to 
recall that for $\eps = 0$ the map $k_0$ is a homeomorphism. Therefore since $k_\eps$ depends 
continuously on $\eps$, when $\eps$ is small enough, the $k$-images of the distinct points $
\zeta_h$, $f_\eps(\zeta_h)$ are distinct. Thus the image of the fundamental segment in $W^u(h)
$ covers a fundamental segment of $\gamma$.

By a similar argument, the same results hold for the stable manifold.
\end{proof}

The main tool we will use in the proof of \Th{Slope} is a theorem stated by Weil in 1935 at the
Moscow Topological Conference \cite{Weil36}, but proved much later by Markley
\cite{Markley69}. Let $L$ be a continuous, semi-infinite, simple
(without self-intersections) curve on the torus $\bT^2 = \bR^2/\bZ^2$ and
$\tilde{L} =\{(x(t),y(t)): t\in \bR^+\}$ be its parameterized lift to the plane.
Let $\rho(x,y)$ denote the standard Euclidean distance from $(x,y)$ to the origin. Then the following theorem holds.

%%%%
\begin{thm}[Weil (1935)]\label{thm:Weil}
If $\rho(x(t),y(t))$ tends to infinity as $t\to \infty$, then $\tilde{L}$ has an
asymptotic direction; that is, either there exists a slope $m^* \in \bR$ such that $\lim_{t\to \infty} y(t)/x(t) = m^*$, or, if this ratio is unbounded, then $\lim_{t\to \infty}x(t)/y(t)= 0$.
\end{thm}
%%%%%

\noindent
In other words, if the lifted curve $(x(t),y(t))$, $t \in [0,\infty)$, has no finite
accumulation points (i.e., there is no sequence $t_n \to \infty$ such that 
$\lim_{n\to \infty}\rho(x(t_n),y(t_n)) = (x_*,y_*)$), then it has an asymptotic direction.

Given these results, we now proceed to prove our theorem:

\begin{proof}[Proof of \Th{Slope}]
Choose lifts $F$ and $K$ of the maps $f$ and $k$ to the plane such that the fixed points $\tilde{h}$
and $\tilde{p}$ of $F$ lie in the unit square $[- \tfrac12,\tfrac12)\times [-
\tfrac12,\tfrac12)$ and that $K(\tilde{h}) = K(\tilde{p}) = O$. We are to prove that the
unstable manifolds $W^u(\tilde{h})$ and $W^u(\tilde{p})$ of $F$ go to infinity and have an
asymptotic direction equal to $\phi^{-1}$---the same as that of the unstable direction of the
linear map $A$. For the first part, we will apply \Th{Weil},
and thus we only need to verify that both of the lifted curves have no finite accumulation points.
We shall prove this for $W^u(\tilde{h})$, since the proof for $W^u(\tilde{p})$ is the same.

Begin by choosing a primary, transverse homoclinic point $q \in W^u(h) \cap W^s(h)$,\footnote
%%%%
{Recall that a homoclinic point $q$ of a plane diffeomorphism with a saddle fixed point $h$ is ``primary," if the closed loop $\overline{W}^u(h,q) \cup \overline{W}^s(h,q)$ has no self-intersection points.}
%%%%%
and consider the segments $W^u(h,q)$ and $W^s(h,q)$, with orientations from $h$ to $q$. We claim that it is possible
to choose $q$ so that $W^s(h,q)$ lies in the interior of the unit square.
Consequently, the forward images, $f^t(q)$, remain in a neighborhood of $h$, moving monotonically from $q$ to $h$ along
the local segment $W^s(h,q)$. Note that the tangent vectors at $q$ to $W^u(h)$ and $W^s(h)$ form a frame. Since $f$ is
symplectic and thus orientation-preserving,
the orientation of this frame is preserved under $Df$.

The existence of such a $q = q_\eps$ follows from the $C^1$
closeness of the manifolds of $f_\eps$ to those of $f_0$, apart from an $\cO(\sqrt{\eps})$-
neighborhood of $O$, and the existence of a transverse homoclinic point on the right-going 
unstable manifold of $W^u(O)$ for $f_0$. Indeed \cite{Lewowicz80} showed that the 
intersections of the stable and unstable manifolds of $f_0$ are transverse everywhere except 
at $O$. Choose one such primary intersection, $r \in W^u(O) \cap W^s(O)$. Since $f_0^{t}(r) 
\to O$ as $t\to \infty$, there is an image $q_0 = f^k(r)$, such its forward images
lie in a ball of radius, say, $\tfrac14$ of $O$.
This homoclinic point $q_0$ is of course still primary and transverse.
Now we take $\eps$ small enough in order that:
(1) $q_0$ is not in an $\cO(\sqrt{\eps})$-neighborhood of $O$;
(2) the intersection point $q_0$ continues to a point $q_\eps$ still in the unit square; and
(3) the intersection at $q_\eps$ remains transverse.
This can be done since since the manifolds of $f_\eps$ are $C^1$ close to those of $f_0$.

The loop $\cC = \overline{W}^u(h,q) \cup \overline{W}^s(h,q)$ is a simple (since $q$ is 
primary) closed
curve on the torus. This curve is not homotopic to zero and has some
nontrivial representation $(m_1,m_2) \in \bZ^2$, in the fundamental group of the
torus. This follows from that fact that $\cC$ is homotopic to the
loop made up from pieces of $W^u(O)$ and $W^s(O)$ of $f_0$, and hence it is
homeomorphic to the related loop of the Anosov map $a$, which is not homotopic
to zero.

A lift of $\cC$ to the covering plane unwinds to an infinite
curve that tends to infinity with rational slope $m_2/m_1$.
The collection of all lifts of $\cC$ cut the plane into infinite number of disjoint strips, 
recall \Fig{Strips_eps06}. Consider the segment $\cU = W^u(\tilde h,\tilde q+m)$ that belongs 
to the upper boundary of one of these strips, say $\Pi_1$. The image $F(\cU) = W^u(\tilde h, 
F(\tilde q+m))$ expands and, by orientation preservation, enters the interior of $\Pi_1$. The 
second endpoint, $F(\tilde q+m) = F(\tilde q) + m'$, where $m' = Am$, lies in the interior of 
the segment $W^s(\tilde h+m',\tilde q+m')$, and is not on the upper boundary of $\Pi_1$, since 
$m'$ is not parallel to $m$.

Two cases are possible. The first occurs when $F(\cU)-\cU$ intersects the boundary of $\Pi_1$ 
only at $F(\tilde q+m)$, i.e., it intersects no other lift of $W^s(h,q)$. The implication is 
that at the next iteration, $F^2(\cU)$ will cross the strip below $\Pi_1$,  due to 
preservation of orientation, etc. In this case, the right-going manifold $W^u(\tilde{h})$ goes to infinity and cannot have accumulation points in finite part of the plane. Similar 
considerations were used in \cite{Grines77}. However, it may be the case that $F(\cU) -\cU$ 
intersects some additional lifts of the segment $W^s(h,q)$ that belong to the boundary of $
\Pi_1$. In this second case, $F(\cU)$ has to leave and then return to $\Pi_1$ since its 
extreme point $F(\tilde q + m)$ still exits $\Pi_1$ through $W^s(\tilde h +m',\tilde q+m')$. 
This leads to a potential problem exemplified by the folds shown in \Fig{ChannelFolds}: there 
could be loops homotopic to zero made up from pieces of $W^u(h)$ and $W^s(h)$.

Nevertheless, as we show in \App{UnBounded}, using specific properties of the map \Eq{StdMap}, 
all points of $W^u(\tilde h)$ tend to infinity. To apply the considerations of 
\App{UnBounded}, we need to choose a fundamental segment $\ell \in W^u(\tilde h)$ in the first 
quadrant, such that for all $(x,y) \in \ell$,
\[
	x + \phi^{-1} y > \tfrac{1}{2\pi}(1+2\eps).
\]
To this end, it is enough to verify that a segment $\ell_0$ exists for $f_0$ (for $\eps=\mu 
=0$), a
fact that is easily numerically verified. Then, $\ell_\eps$ satisfies the requirement for 
$f_\eps$ for small enough $\eps$, since the unstable manifolds of both $\tilde h$ and 
$\tilde p$ are $C^1$-close to those of $f_0$ on compact sets away from an 
$\cO(\sqrt{\eps})$-neighborhood of $O$ (in fact, $C^0$-closeness of manifolds is sufficient).

Therefore we have shown that there are no accumulation points, and \Th{Weil} applies, and thus $W^u(\tilde{h})$ has a limiting slope. The same argument applies to $W^u(\tilde{p})$, since it too is mapped onto $\Gamma$ by $K$.

Finally we need to prove that the limiting slope is indeed equal to $\phi^{-1}$. To that end
we use \Eq{StdMap} with the assumption that $g$ has degree one:
\beq{DegreeOne}
	g(x) = x + \hat{g}(x),
\eeq
where $\hat{g}(x)$ is a continuous, periodic function. If $(x_h(s),y_h(s)) = \zeta_h(s)$, then
under the map $F(x,y) = \zeta_h(s+1)$. Define the slope of the chord from the fixed point to $
\zeta_h(s)$ by
\[
	m(s) = \frac{y_h(s)-y_h}{x_h(s)-x_h}.
\]
Subtracting the fixed point from both sides of \Eq{StdMap}, and computing the slope gives,
after some algebra,
\beq{SlopeIterate}
	m(s+1) = \frac{1+ m(s)+\frac{\hat{g}(x_h(s)) -\hat{g}(x_h)}
			{x_h(s)-x_h}}
			{2+m(s)+\frac{\hat{g}(x_h(s))-\hat{g}(x_h)}{x_h(s)-x_h}}.
\eeq
Now, according to \Th{Weil}, the slope $m(s)$ has a limit,
$m^*$. Moreover, since $x_h(s)$ is unbounded and $\hat{g}(x)$ is periodic,
\[
	\lim_{s\to\infty}\frac{\hat{g}(x_h(s))-\hat{g}(x_h)}{x_h(s)-x_h} = 0 .
\]
Thus after taking the limit on both sides of \Eq{SlopeIterate} we come to
\[
	m^*=\frac{1+m^*}{2+m^*},
\]
which implies, since $m(s) > 0$, that $m^*= \phi^{-1}$.
\end{proof}

Not only do $W^u(\tilde{h})$ and $W^u(\tilde{p})$ for the lift $F$ have the same limiting
slope, as implied by \Th{Slope}, but they converge to each other. This can be seen numerically
by choosing the first parameter value for which each curve crosses a particular abscissa value
$\xi = x_p(s_p) = x_h(s_h)$. Let $\eta_h(\xi) = y_h(s_h)$, and $\eta_p(\xi) = y_p(s_p)$ denote
the corresponding ordinates. We observe (again computing up to $\xi = 1000$) that the vertical
distance between these curves decreases algebraically as
\beq{VerticalDistance}
	\eta_h(\xi) - \eta_p(\xi) = \cO\left(\xi^{-1}\right),
\eeq
see \Fig{DeltaYvsEps}. A proof of this result appears to be nontrivial.
Indeed, this decay is not uniform---since we define $\xi$ to be the
first horizontal crossing, the formation of folds causes the vertical distance to exhibit jumps.
More generally, close approaches to the stable channel cause oscillations;
the first place this occurs is near $\xi = 2.5$. In this
case the vertical distance decreases monotonically up to $\xi = 2$, it subsequently increases
as the unstable channel crosses the initial segment of the stable channel, reaching a local
maximum near $(x_h,y_h) \approx (2.59,1.64)$. The next two local maxima occur at $(5.39,3.38)$
and $(7.23,4.57)$, again correlated with crossing the stable channel. As can be seen in
\Fig{DeltaYvsEps}, these local maxima occur at approximately the same values of $\xi$ for any
value of $\eps$. Indeed, the function $[\eta_h(\xi)-\eta_p(\xi)]\xi$ appears to be
quasiperiodic, with two dominant periods $\Delta \xi = 2.61$ and $6.90$, again independent of 
$\eps$.

%%%%%
\InsertFig{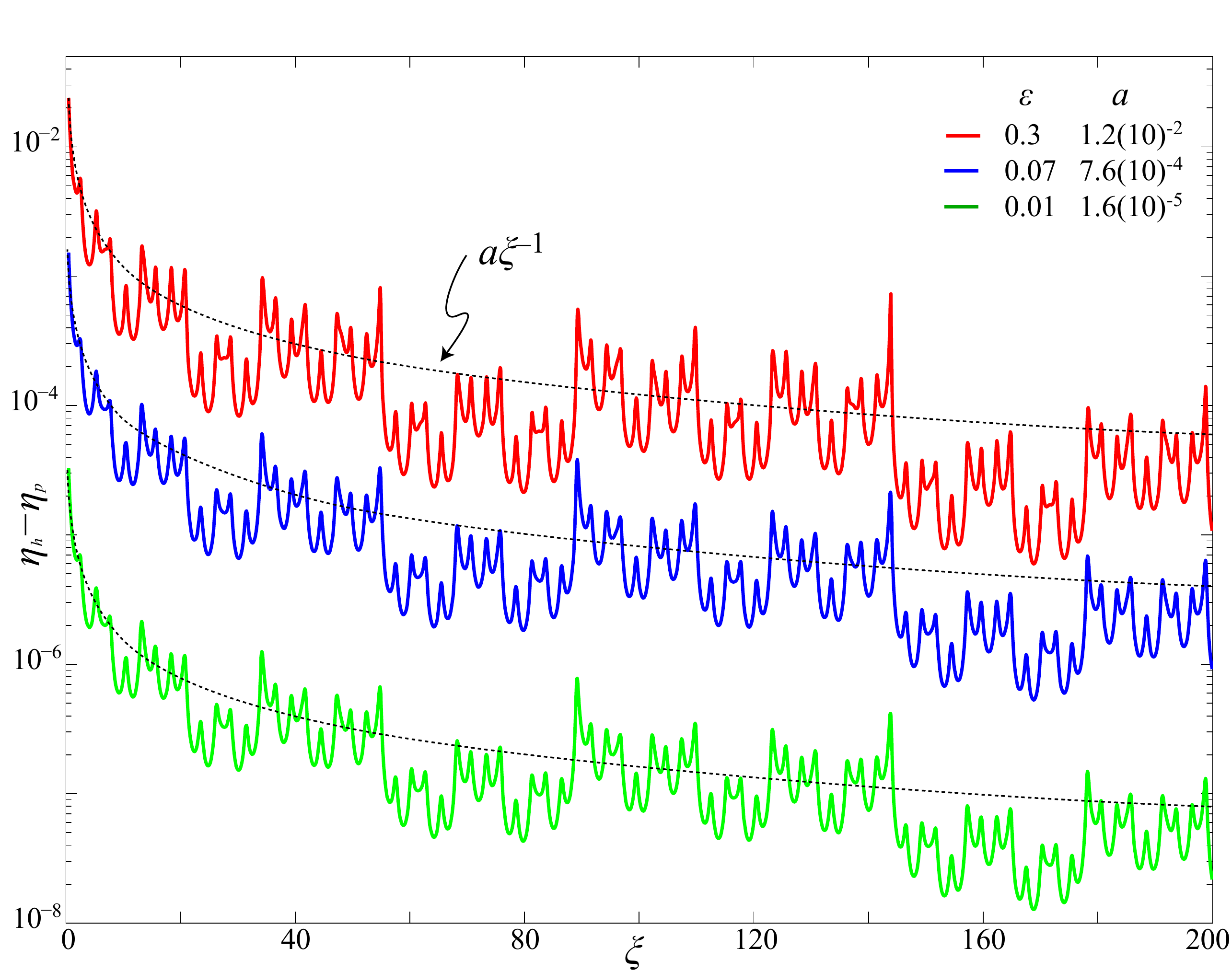}{Vertical distance between the unstable manifolds of the saddle and parabolic
fixed points for the lift of \Eq{StdMap} up to $\xi = 200$ for three values of $\eps$.
Also shown (dashed curves) are the graphs $a/\xi$, with the values of $a$ as shown.}
{DeltaYvsEps}{4in}
%%%%%

%%%%%%%%%%%%%%%%%
%% Channel Area
%%%%%%%%%%%%%%%%%
\section{Channel Area}\label{sec:ChannelArea}
%%%%%%%%%%%
Our goal in this section is to compute the area contained in the channel $E$ ``between" the
invariant manifolds of the parabolic and hyperbolic points, and to show that, when $\eps$ is
small, the total area of the channel is less than one, implying that the dynamics of $f$ is
partitioned into two invariant regions of nonzero measure. As noted above, and in particular
in \Fig{TwoManifoldsE0379}, it is numerically infeasible to simply iterate chosen initial
conditions in the channel, since these numerical trajectories rapidly fill every pixel in the
image. Thus, instead, we will use numerical computations of the stable and unstable manifolds 
that form the boundaries of $E$.

For the lift $F$, the region ``between" the curves $W^u(\tilde h)$ and $W^u(\tilde p)$
corresponds to the upper, unstable half, $\tilde E^+$, of the lifted channel, e.g., the gray
region in \Fig{Strips_eps06}. Its area
can be easily computed up to some finite extension $\xi$ on the plane. Let $\tilde E^+(\xi)
\subset \bR^2$ denote the region with boundary
\beq{ChannelBoundary}
	\partial \tilde E^+(\xi) = \{(x,0): x_h\le x \le x_p\}
	  + W^u(\tilde p,\zeta_p) + V_\xi - W^u(\tilde h,\zeta_h)
\eeq
where $\zeta_{h,p} = (\xi,\eta_{h,p}(\xi))$ are points on the unstable manifolds, and
\beq{VerticalSegment}
	V_\xi = \{(\xi,y): \eta_p(\xi) \le y \le \eta_h(\xi)\}
\eeq
is the connecting vertical segment. Let
\[
	\tilde A^+(\xi) = \mbox{Area}(\tilde E^+(\xi))
\]
denote the area of the unstable channel up to $\xi$. This can be most easily computed using
the relation between action and area, see \App{Action}. The results are shown in
\Fig{LiftedChannelArea} for several values of the cut-off $\xi$, as a function of $\eps$.

Since, by \Eq{ParabolicFixedPt} and \Eq{SaddleFixedPt}, $x_p-x_h = \cO(\eps^{1/2})$, and the
slope of the unstable eigenvector of \Eq{Jacobian} at $x_h$ is $\cO(\eps^{1/2})$, it can be
seen that $\tilde A^+(\xi) = \cO(\eps^{3/2})$. More precisely, this follows from the 
Hamiltonian normal form valid near $\eps=\mu = 0$. This asymptotics is supported by the
calculations shown in \Fig{LiftedChannelArea}.

%This is derived using the local approximation of the diffeomorphisms in the unfolding (with
%parameters $\eps, \mu$) by an unfolding of Hamiltonian vector fields near the degenerate
%saddle of the Hamiltonian vector field \cite{Takens74} (this unfolding can be reduced to the
%form \Eq{HamNormalForm}, see \Sec{HamiltonianFlow})
%\beq{HamNormalForm}
%	H=\tfrac12{y^2}-\mu x+\tfrac12{\eps x^2}-\tfrac14 {x^4}
%\eeq
%along the bifurcation curve where $\mu_p = \mu = 2(\eps/3)^{3/2} +\cdots$.

%%%%%
\InsertFig{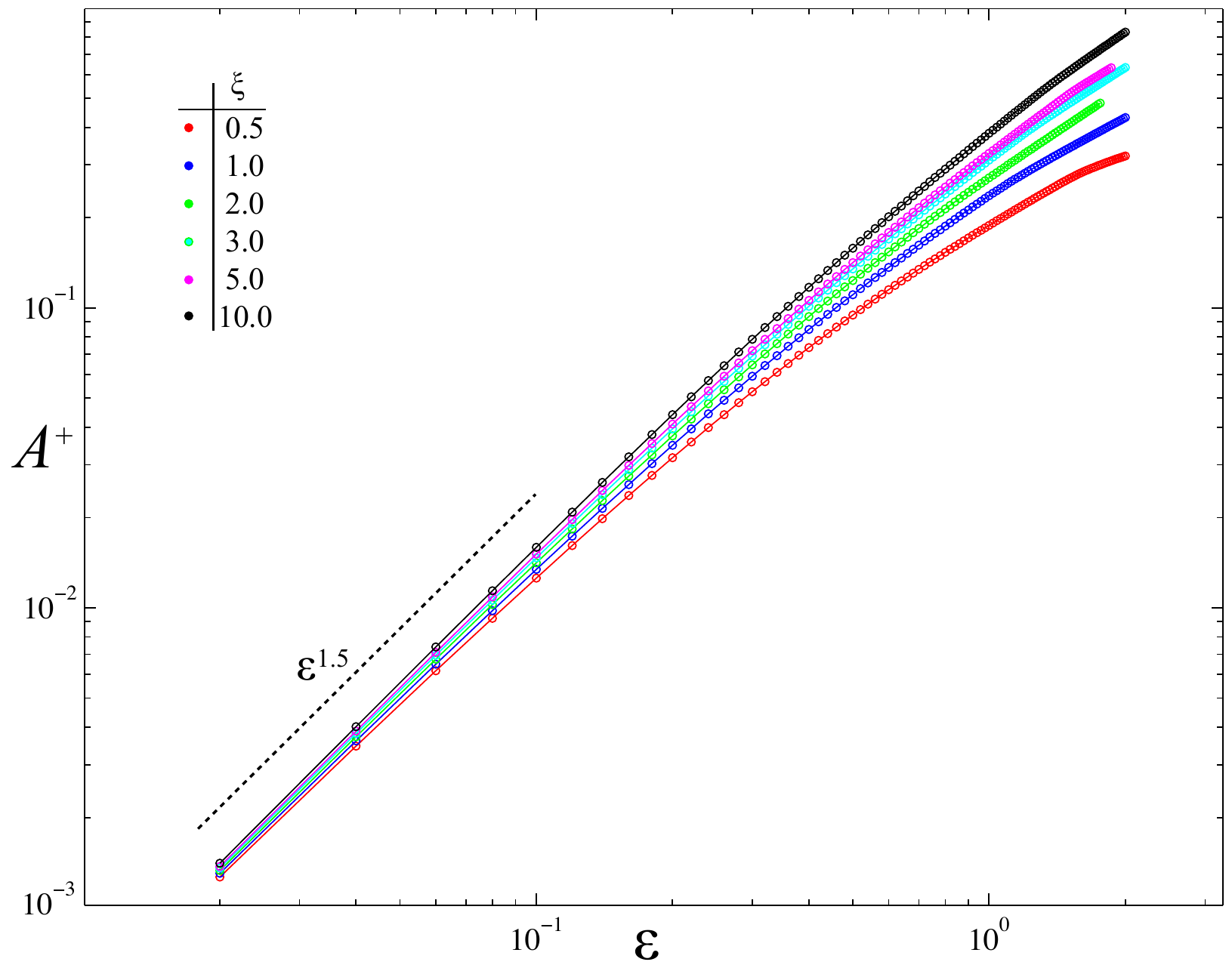}{Area of the lift of the upper, unstable half of the channel for
the lift $F$ of \Eq{StdMap} as a function of $\eps$ for six values of the cut-off $\xi$.}
{LiftedChannelArea}
{4in}
%%%%%

Note that $\tilde A^+_\eps(\xi)$ must be unbounded as $\xi \to \infty$: this is a consequence
of area preservation. Indeed, the unstable channel can be generated by iteration of a
``fundamental domain." For each point $\zeta = (\xi,\eta)$, on a branch of an unstable
manifold, the segment $\ell = W^u(\zeta,F(\zeta))$ generates the entire manifold under
iteration. A fundamental domain, $D$, for the unstable half of the channel corresponds to
the region with boundary
\[
	\partial D = W^u(\zeta_p,F(\zeta_p)) + F(V_\xi) - W^u(\zeta_h,F(\zeta_h)) - V_\xi ,
\]
where $\zeta_{h,p} = (\xi,\eta_{h,p})$ are points on the respective manifolds, recall
\Fig{ChannelFolds}. Note that the image $F(V_\xi)$ of the vertical segment
\Eq{VerticalSegment} is a line segment with unit slope since \Eq{StdMap} has constant twist.
The channel $\tilde E^+$ is generated by the images of $D$, and thus,
with each iteration of the map, its area grows by the area of the fundamental domain.

Since $\tilde A^+(\xi) \to \infty$ as $\xi \to \infty$, the vertical distance between the
manifolds, \Eq{VerticalDistance}, cannot decrease more rapidly than $\xi^{-1}$, confirming the
decay observed in \Fig{DeltaYvsEps}. Given this rate of convergence, it is clear that $A^+$
must increase logarithmically with the intercept $\xi$, and this is confirmed by the
computations in \Fig{AreavsXE01-05}. The oscillations seen in
this figure correspond to those seen in the vertical distance in \Fig{DeltaYvsEps}.
%%%%%
\InsertFig{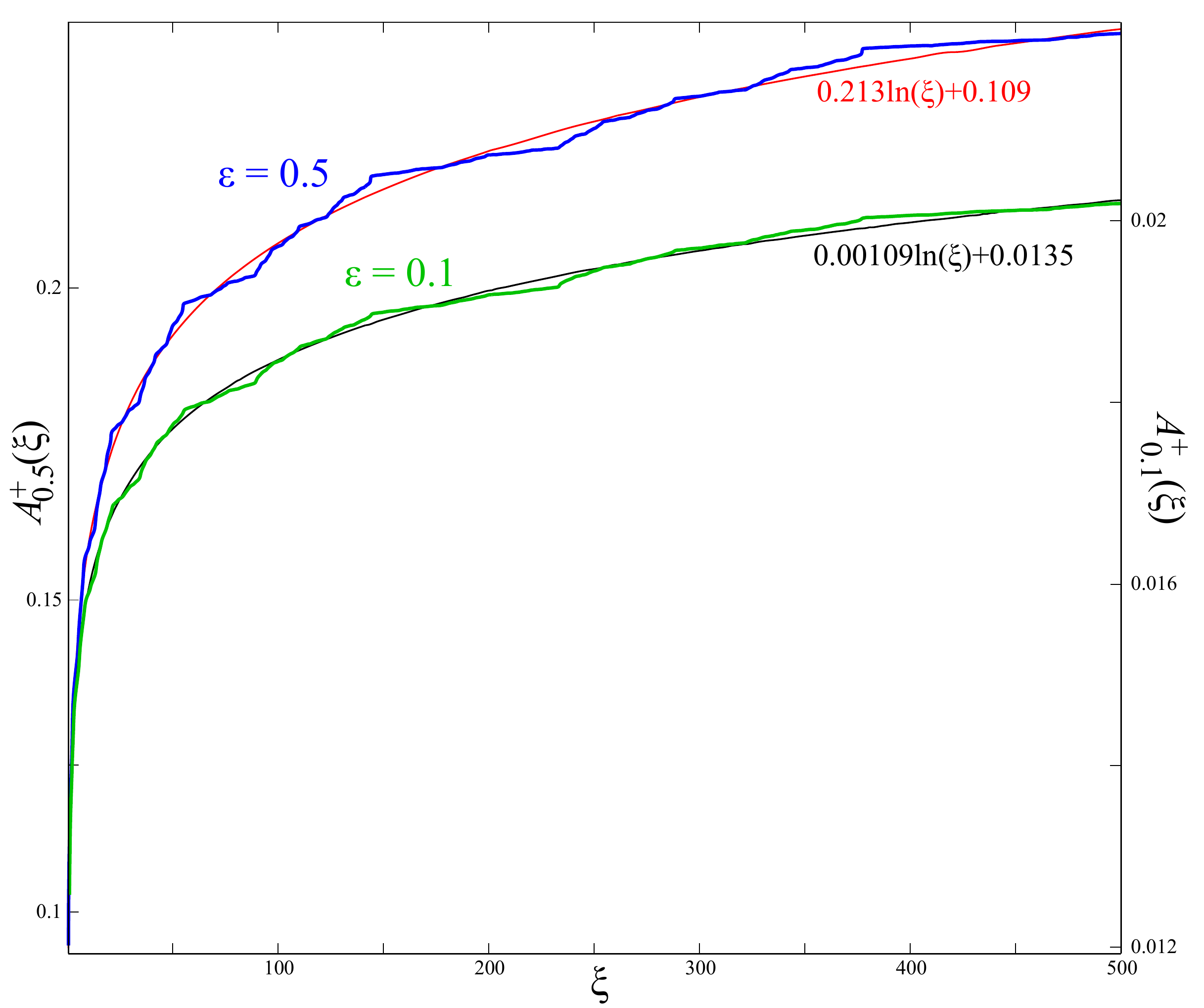}{Area $A^+$ of the upper channel on the plane as a function of
the horizontal extent, $\xi$ of the channel for $\eps=0.5$ (blue, left axis) and $\eps = 0.1$ (green, right axis).
Fits (black) are to the log functions shown
using data up to $\xi = 1000$.}{AreavsXE01-05}{4in}
%%%%%

So, how can we conjecture that the projection $E$ of the channel $\tilde E$ onto the torus has
finite area? This must happen by the creation of heteroclinic orbits that lead to re-injection
of the channel, as we noted in \Sec{InvariantManifolds}, and showed in
\Fig{ChannelFolds}. The region $R$ in this figure, and all of its forward images, are inside
the channel, and their areas are be deleted from the lifted area of the channel upon
projection.

To compute the area of $E = E^+ \cup E^-$ on $\bT^2$ accounting for the overlap of the channel
with itself, we resort to an image-based calculation. To start, the region $\tilde E^+$ is
computed up to an extension $\xi$ as before. This region is projected into $\bT^2$, and
discretized into an $N\times N$ pixel image---a pixel is deemed to be occupied if there is a
point on the manifolds $W^u(h)$ or $W^u(p)$ that lands in the pixel. The region is computed by
filling the pixels vertically from $\eta_p$ to $\eta_h$. We fill in the channel sequentially,
increasing the cut-off $\xi$; in this case the folds, which are in the interior of the
channel, do not cause a problem with the filling algorithm. Finally the full channel is
computed by applying the reflection $S$ to the pixels in the unstable channel, giving the
image
$E_{N \times N}$. An example, for $\eps = 0.5$, is shown in \Fig{PixelChannel}.

%%%%%
\InsertFig{Cmatrix_E50_Cm100_nP4000}{Discretized channel $E_{N\times N}$ for $\eps = 0.5$ with
$
\xi = 100$ and $N = 4000$. The channel (black region) intersects $7,040,240$ pixels, or
$44.00\%$ of the area. For $ \xi=100$, $\eta_h-\eta_p = 1.47(10)^{-4}$, so that the channel
height is less then one pixel. Note that the folds in the manifolds are hidden
since they occur in the interior of the channel, recall
\Fig{FourManifoldsE05}.
%This width spikes to $4.51(10)^{-3}$ near $\xi = 143$, but
%appears to otherwise never exceed one pixel
}{PixelChannel}{6in}
%%%%%

We observe that as the number of pixels, $N$, grows, the computed channel area monotonically
decreases, and that the error is proportional to $N^{-1}$, see \Tbl{ChannelArea}. We can use
this to extrapolate to get an estimate of the area to an absolute error less than $10^{-4}$,
the column labeled $b_N$ in the table. A final extrapolation to remove errors proportional to
$N^{-2}$, the column $c_N$, reduces the error estimate slightly.
Thus we estimate that
\[
	A_{0.1}(100) = 0.03679 \pm 2(10)^{-5}.
\]
Note that the area of the upper lifted channel (computed using the action) is $\tilde A^+_{0.1}(100) = 0.018563$, which when doubled gives a total channel area of $0.037126$. Thus the fraction of area excluded due to overlap is about $0.9\%$.

\begin{table}[ht]
  \centering
  \begin{tabular}{@{} r|lll @{}} % Column formatting, @{} suppresses leading/trailing space
	   \thdBar{$N$}  &\thd{$a_N$} & \thd{$b_N$} & \thd{$c_N$}\\
  \hline
	500		&0.486252		\\
	1000	&0.289216	&0.092180	\\
	2000	&0.167454	&0.045692	&0.030196\\
	4000	&0.103127	&0.038800	&0.036503\\
	8000	&0.070165	&0.037203	&0.036671\\
	16000	&0.053519	&0.036873	&0.036763\\
	32000	&0.045162	&0.036805	&0.036782\\
	64000	&0.040973	&0.036784	&0.036777\\
	128000	&0.038890	&0.036807	&0.036815\\
	\hline
  \end{tabular}
  \caption{Area of the discretized channel, $a_N$, for $\eps = 0.1$ and $\xi = 100$
  as a function of the number of pixels, $N\times N$,
  in the image. The extrapolation, $b_N = 2a_N-a_{N/2}$, removes errors $\cO(N^{-1})$
  and the second, removing errors $\cO(N^{-2})$, is $c_N = \tfrac13(4b_N-b_{N/2})$.
  }
  \label{tbl:ChannelArea}
\end{table}

After this extrapolation, we vary the cut-off, $\xi$, to attempt to estimate
$A_\eps = \lim_{\xi \to \infty} \mbox{Area}(E_\eps(\xi))$. The results, again for $\eps = 0.1$ are shown in \Tbl{AreaVsXi}. After the second extrapolation, we estimate that the true area of the channel is
\[
	A_{0.1} = 0.03990 \pm 5(10)^{-5} .
\]
Using these ideas, we compute the area as a function of $\eps$ for three values of the channel
cut-off, $\xi$, see \Fig{AreaVsEps}. For $\eps < 0.01$, the results have not converged: they
depend on $\xi$ significantly. It is strange that for these values, $A(200) > A(300)$;
this is due to error in the extrapolations for $b_N$---none of the computed values $a_N(\xi)$
have this contradictory property. The error bars in the figure are estimated by
$|b_{32000}-b_{16000}|$. When $\eps >0.01$, the area seems to have converged
with $\xi = 300$. As in \Fig{LiftedChannelArea} the area again grows as $
\eps^{3/2}$.
A fit over the interval $0.01<\eps<1.0$ gives
\beq{AreaOfEpsilon}
	A_\eps = (1.02\pm0.09) \eps^{1.42 \pm 0.07},
\eeq
while a fit over the narrower interval $0.04<\eps<0.6$ gives an exponent of $1.49 \pm 0.04$.
When $\eps$ approaches $1$, the power law predicts that
$A_\eps \to 1$, and, as can be seen in the figure, the area saturates at one.

\begin{table}[htdp]
\centering
\begin{tabular}{@{} r|lll|ll|l @{}} % Column formatting, @{} suppresses leading/trailing space
\thdBar{$\xi$} &\thd{$a_{8000}$} & \thd{$a_{16000}$} & \thdBar{$a_{32000}$} &
\thd{$b_{16000}$} & \thdBar{$b_{32000}$} & \thd{$c_{32000}$}\\
\hline
%10	& 0.035029	& 0.033416	& 0.032615	& 0.031803	& 0.031814	& 0.031818\\
30	& 0.043929	& 0.039000	& 0.036549	& 0.034071	& 0.034098	& 0.034107\\
%50	& 0.051924	& 0.043669	& 0.039563	& 0.035414	& 0.035457	& 0.035471\\
70	& 0.059595	& 0.047852	& 0.041964	& 0.036109	& 0.036076	& 0.036065\\
%90	& 0.066845	& 0.051682	& 0.044074	& 0.036519	& 0.036466	& 0.036448\\
110	& 0.073580	& 0.055382	& 0.046242	& 0.037184	& 0.037102	& 0.037075\\
%130	& 0.080382	& 0.058951	& 0.048211	& 0.037520	& 0.037471	& 0.037455\\
150	& 0.087464	& 0.062759	& 0.050329	& 0.038054	& 0.037899	& 0.037847\\
%170	& 0.095377	& 0.067086	& 0.052803	& 0.038795	& 0.038520	& 0.038428\\
190	& 0.102100	& 0.070515	& 0.054570	& 0.038930	& 0.038625	& 0.038523\\
%210	& 0.109032	& 0.074123	& 0.056471	& 0.039214	& 0.038819	& 0.038687\\
230	& 0.116212	& 0.077799	& 0.058356	& 0.039386	& 0.038913	& 0.038755\\
%250	& 0.122361	& 0.081117	& 0.060222	& 0.039873	& 0.039327	& 0.039145\\
270	& 0.128380	& 0.084259	& 0.061867	& 0.040138	& 0.039475	& 0.039254\\
%290	& 0.134881	& 0.087683	& 0.063688	& 0.040485	& 0.039693	& 0.039429\\
310	& 0.141428	& 0.091035	& 0.065402	& 0.040642	& 0.039769	& 0.039478\\
%330	& 0.147911	& 0.094438	& 0.067184	& 0.040965	& 0.039930	& 0.039585\\
350	& 0.154256	& 0.097727	& 0.068905	& 0.041198	& 0.040083	& 0.039711\\
%370	& 0.160689	& 0.101105	& 0.070673	& 0.041521	& 0.040241	& 0.039814\\
390	& 0.167162	& 0.104657	& 0.072596	& 0.042152	& 0.040535	& 0.039996\\
%410	& 0.173186	& 0.108444	& 0.074678	& 0.043702	& 0.040912	& 0.039982\\
430	& 0.180157	& 0.112024	& 0.076516	& 0.043891	& 0.041008	& 0.040047\\
%450	& 0.186159	& 0.115925	& 0.078686	& 0.045691	& 0.041447	& 0.040032\\
470	& 0.192669	& 0.119685	& 0.080630	& 0.046701	& 0.041575	& 0.039866\\
%490	& 0.198778	& 0.122788	& 0.082201	& 0.046798	& 0.041614	& 0.039886\\
510	& 0.205416	& 0.126197	& 0.083939	& 0.046978	& 0.041681	& 0.039915\\
\end{tabular}
\caption{Channel area for $\eps = 0.1$ as a function of the cut-off $\xi$.
Columns labeled $a_N$ are the computed areas for $N\times N$ pixels. The final three
columns show extrapolations $b_{N} = 2a_{N}-a_{N/2}$,
and $c_{N} = \tfrac13(4b_{N}-b_{N/2})$.}
\label{tbl:AreaVsXi}
\end{table}%

%%%%%
\InsertFig{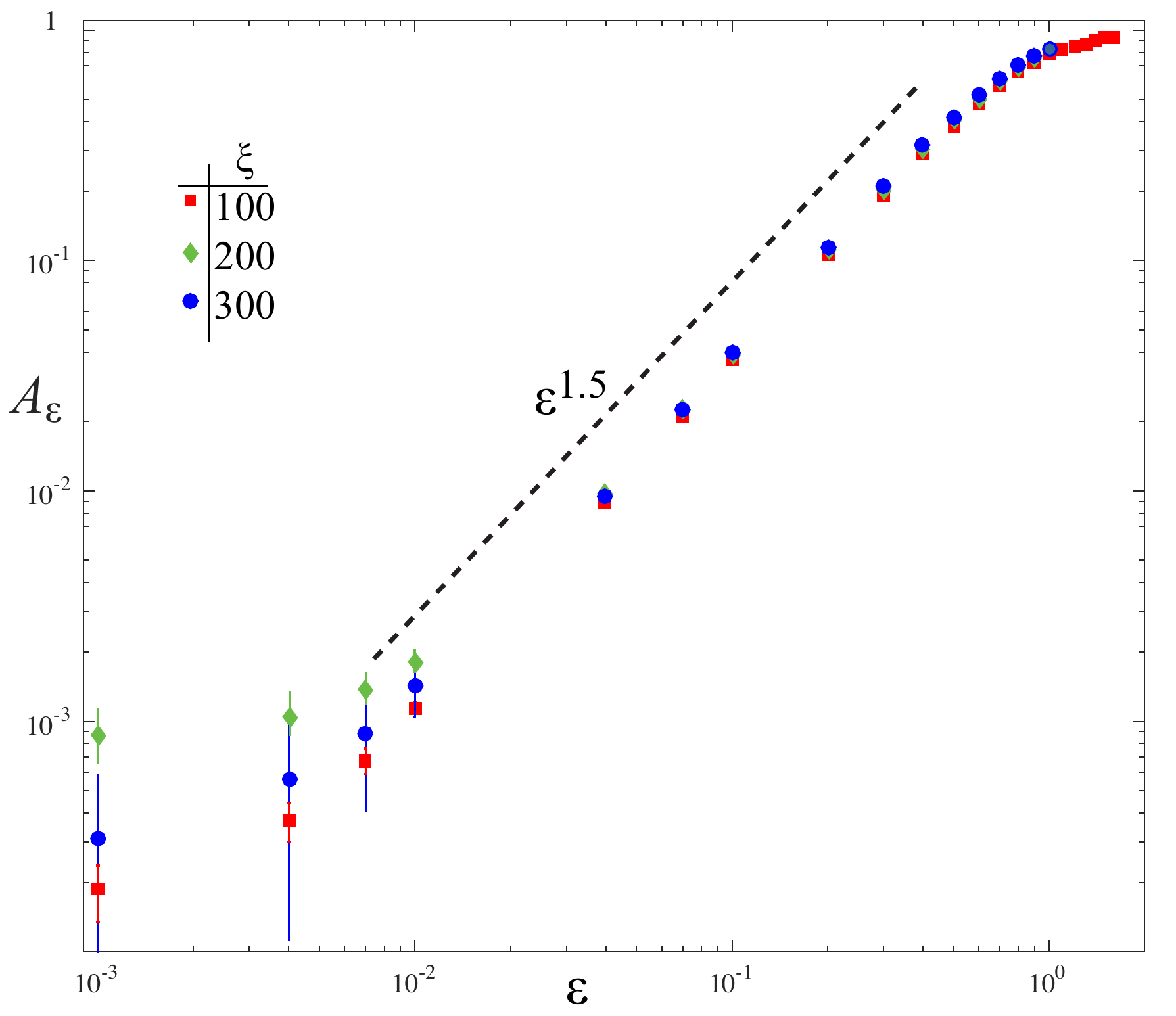}{Area of the channel $E$ using the second order extrapolation as in
\Tbl{ChannelArea}, for $N=32,000$, as a function of $\eps$ for three values of the channel
cut-off, $\xi$ as shown. Error, estimated as the difference $|b_{32000}-b_{16000}|$, is only
visible when $\eps \le 0.01$}{AreaVsEps}{4in}
%%%%%
%%%%%%%%%%%%%%%%%
%% Lyapunov Exponents
%%%%%%%%%%%%%%%%%
\section{Lyapunov Exponents}\label{sec:Lyapunov}

The Lyapunov exponents of the family $f_\eps$ appear, by the standard computation, to be
positive. We compute the finite-time exponent
\beq{FTLE}
	\lambda_\eps(x,y,T) = \frac{1}{T} \ln \|Df_\eps^T(x,y) v_0\|, \quad
	v_0 = \begin{pmatrix} 0 \\ 1 \end{pmatrix}
\eeq
for an initial condition $(x,y)$ with the vertical initial deviation vector $v_0$ over a time
$T$. The results shown in \Fig{LyapunovVsEps} give the mean exponent for $400$ initial
conditions with $T=10^4$ (for these parameters standard deviation of the distribution of
exponents is smaller than $0.005$). Note that
$\langle \lambda_0\rangle \approx 0.902177 < \ln \phi^2 \approx 0.9642$,
the exponent of the Anosov map \Eq{Anosov}
(i.e., $\eps = -1$ and $\mu=0$). The exponent decreases
monotonically from its value at $\eps = 0$ until $\eps = 1.55$, when it begins to increase
(though not monotonically), reaching $\langle \lambda \rangle \approx 1.1$ at $ \eps = 5$.
%Recall that the large, period-$5$ elliptic island, \Sec{Elliptic} was born at $\eps = 1.16$

%%%%%
\InsertFig{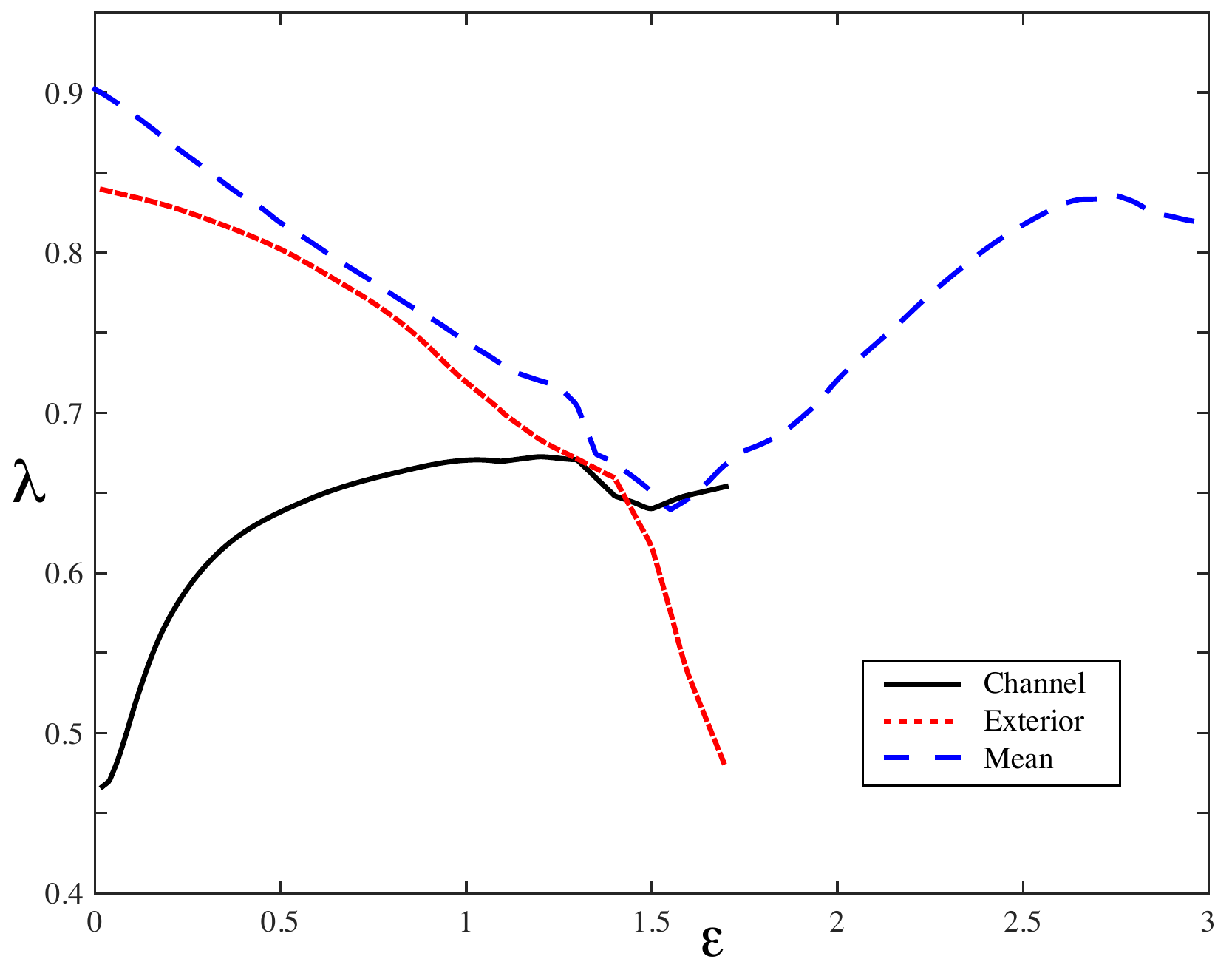}{Lyapunov exponents for the map \Eq{StdMap} with \Eq{Force} as a
function of $\eps$. The dashed curve (blue) is the average exponent for a grid of $20\times
20$ initial conditions, each iterated $10^4$ steps. The dotted (red) and solid (black)
curves show separately the mean exponents for orbits in the exterior and interior of the 
channel, estimated using $3000\times 3000$ pixel image with a cut-off $\xi = 50$.}
{LyapunovVsEps}{4in}

To estimate the exponent separately for orbits that lie in the channel $E$ and orbits
that lie in its exterior, $H = \bT^2 \setminus E$, we use the $N\times N$-pixel
approximation of the channel, $E_{N\times N}(\xi)$, recall \Fig{PixelChannel}.
Initial conditions for the exterior computation
are chosen in each pixel of $H_{N \times N}(\xi) = \bT^2 \setminus E_{N\times N}(\xi)$, and
each is iterated only over the time that it remains in $H_{N\times N}(\xi)$: the time $T$ in 
\Eq{FTLE} is chosen so that the orbit segment from $0$ to $T$ lies in $H$
Similarly, an in-channel, Lyapunov exponent can be computed by choosing initial conditions in
$E_{N\times N}(\xi)$, iterating them only as long as they remain in the approximate channel.
The resulting finite-time Lyapunov exponents are shown, for $\eps = 0.1$,
as a function of their initial condition in \Fig{LyapunovChannel} using a channel cut-off of $
\xi = 50$. The mean exponent for initial conditions in $H$ is $\langle \lambda_{0.1} \rangle_H 
= 0.8352$, while $\langle \lambda_{0.1} \rangle_E = 0.5123$.

The mean exponents in $H$ and $E$ are also shown in \Fig{LyapunovVsEps} as a function of $\eps
$. Since these computations are for a fixed number of pixels, $N = 3000$, the approximation 
$H_{N\times N}$ will vanish for large enough $\eps$. Since orbits leave this gridded 
approximation rapidly, we do not show these curves for $\eps > 1.6$.
Note that the exponent for initial conditions in $H$ is a monotonically decreasing function of $\eps$, while that for $E$ primarily increases. It appears that the minimum of the globally averaged exponent (dashed curve in the figure) corresponds to the point at which the channel area reaches $\cO(1)$ so that the essentially all orbits are in the channel. In principle, the globally averaged exponent should be the weighted average of the channel and exterior results---but this is not true for the figure. The reason is that the computations are carried out over different time intervals. The latter two are averages over the shorter time during which orbit segments remain in $E$ or in $H$. We observe that the value of \Eq{FTLE} increases with $T$; the result is that both $\langle \lambda_{0.1} \rangle_E$ and $\langle \lambda_{0.1} \rangle_H$ are smaller than those of the global average, which used $T=10^4$.

%%%%%
\InsertFig{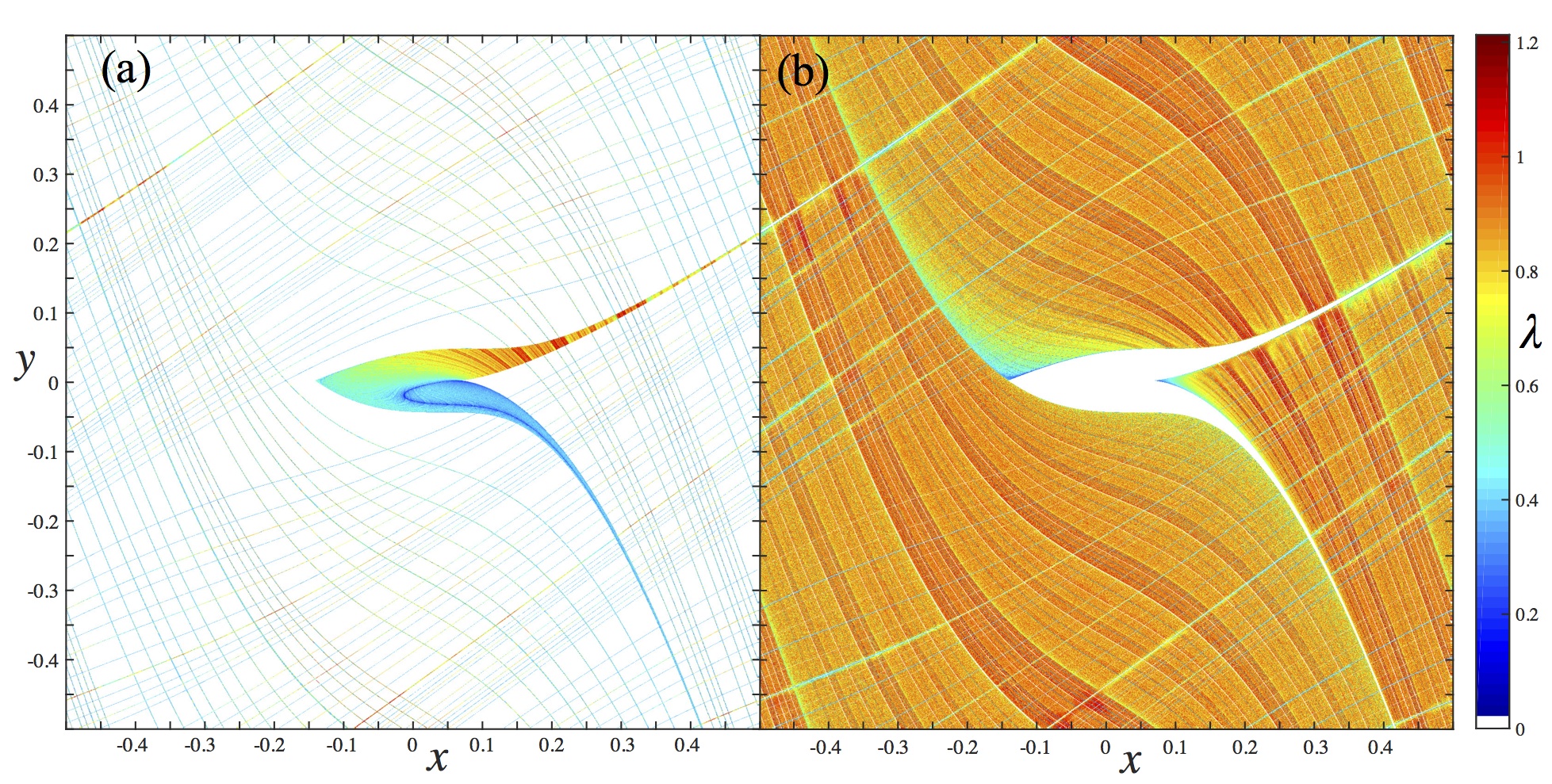}{Distribution of the finite-time Lyapunov exponents
\Eq{FTLE} for for $\eps = 0.1$. Panel (a) shows initial conditions in the channel and (b) in
the exterior of the channel. The channel is computed up to the cut-off $\xi = 50$, and
discretized onto a $3000\times 3000$ grid.}{LyapunovChannel}{7.5in}
%%%%%

%%%%%%%%%%%%%%%%%
%% Section Name
%%%%%%%%%%%%%%%%%
\section{Conclusions}\label{sec:Conclusions}
We have provided numerical evidence for the three conjectures of \Sec{Introduction} for a
family of parabolic standard maps $f_\eps: \bT^2 \to \bT^2$, \Eq{StdMap} with force
\Eq{Force}, that are homotopic to the Anosov map \Eq{Anosov}, but which have a pair of fixed
points for each $\eps >0$, one hyperbolic and one parabolic. We showed that the
right-going stable and unstable manifolds of these fixed points bound a channel
$E \subset \bT^2$.
The lift $\tilde{E}$ of the channel to the plane has unstable boundaries that are asymptotic
to lines of slope $\phi^{-1}$, the slope of the unstable manifolds of the Anosov map. Since
these maps are, in addition, reversible, the same assertion concerning the slope is valid for
stable manifolds. The height of the lifted channel approaches zero as $x^{-1}$, which is the
maximal rate consistent with area-preservation.
\begin{itemize}
	\item We have computed the area $A_\eps(E)$ for the lift using the action, and on the
torus using pixel-based computations. We show that $A_\eps(E) < 1$ when $\eps < 1$. We
conjecture that there is a transition near $\eps = 1$ where the measure of the channel reaches
one.	
	\item We have found elliptic periodic orbits in the channel for several values of $\eps$.
These are formed through saddle-center bifurcations near tangencies of the stable and unstable
manifolds of the hyperbolic point, i.e., by the Newhouse mechanism. We conjecture that there
are elliptic orbits in the channel for arbitrarily small, positive $\eps$, and that there are
no elliptic orbits in its complement, $H$.
	\item We have computed finite-time Lyapunov exponents for orbit segments both in the
channel $E$ and in its complement, $H$. As $\eps \to 0^+$ it appears that the former
monotonically decrease, while the latter limit to the exponent of the almost hyperbolic map
$f_0$. This occurs even though a naive numerical iteration of any given initial condition
appears to fill every pixel of a computed image.
	\end{itemize}
We hope that these results will present convincing arguments in favor of
the hypothesis that a generic, sufficiently smooth symplectic diffeomorphism does
have a positive measure invariant set where its Lyapunov exponent is
positive and that is it non-uniformly hyperbolic on this set. This would show
the drastic difference between properties of sufficiently smooth and
$C^1$-smooth symplectic diffeomorphisms where a generic case is zero
Lyapunov exponent almost everywhere with respect to the Lebesgue measure \cite{Bochi02}.
%%%%%%%%%%%%%%%Appendices%%%%%%%%%%%%
\pagebreak
\section*{Appendices}
\appendix

%%%%%%%%%%%%%%%%%
%% Parabolic Manifolds: Polynomial Approximation
%%%%%%%%%%%%%%%%%
\section{Parabolic Manifolds}\label{app:ParabolicManifolds}

As shown by \cite{Fontich99}, a map of the form \Eq{StdMap}, with a parabolic fixed point at $
(x_p,0)$ such that $g(x-x_p) = \cO((x-x_p)^2)$ has
a pair of stable and unstable manifolds. In the neighborhood of the fixed point, these can be
parametrically represented as
\beq{ParabolicApprox}
	W^u(p) = \begin{pmatrix} x_p\\0\end{pmatrix} +
	  \begin{pmatrix} s^2 \\ \alpha_3 s^3+ \alpha_4 s^4 + \alpha_5 s^5 + \ldots
\end{pmatrix}
\eeq
under the assumption that the dynamics on the manifold is parameterized by the one-dimensional
map $\sigma: \bR \to \bR$,
\[
	s \mapsto \sigma(s) = s + \beta_2 s^2 + \beta_3 s^3 + \ldots
\]
Demanding that this set be invariant gives a set of equations that can be solved, 
order-by-order, for the coefficients $\{\alpha_i, \beta_j\}$. For the case \Eq{Force}, the 
result is
\bsplit{ParabolicManifold}
  \alpha_3 &= \sqrt{\frac{2\pi k}{3}}, \quad
 &\alpha_4 = -\frac{\pi k}{2}, \quad
 &\alpha_5 = \sqrt{\frac{6 \pi^3}{k}} \frac{12 + 11 k^2}{144},\\
  \beta_2 &= \sqrt{\frac{\pi k}{6}}, \quad
 &\beta_3 = \frac{\pi k}{6}, \quad
 &\beta_4 = \sqrt{\frac{6 \pi^3}{k}} \frac{4 + k^2}{96},
\esplit
where $k \equiv \sqrt{\eps(2+\eps)}$. These expansions are well-defined only for $k \neq 0$,
requiring $\eps \neq 0$. Note that since $\alpha_3 > 0$, $W^{s,u}(p)$ has the form of a cubic
cusp.

This expansion, while useful for small $s$, does not give a good representation too far from
the fixed point. For example, the degree-$10$ polynomial approximations are compared with the
numerically generated manifolds of $(x_p,0)$ in \Fig{ParabolicManifolds}.

\section{Creation of elliptic points from tangency of fixed point sets}\label{app:creation_elliptic}

In this appendix we present a justification of the method of finding elliptic points used in 
\Sec{EllipticDynamics}. We consider only the case of an $S$-reversible area-preserving map, $f \circ S = S\circ f^{-1}$ for which the involution $S$ has a smooth line of fixed points, $\Fix{S}$.

\begin{thm} Suppose that $f$ is a $C^2$ area-preserving
diffeomorphism that is reversible w.r.t. a smooth involution $S$, and
the set $\Fix{S}$ of the involution fixed points is a smooth curve. Then if
$\xi = \Fix{S} \cap f^p(\Fix{S})$ is a point of transversal
intersection, it is a point on either an elliptic or a hyperbolic period-$2p$ orbit,
while if $\xi$ is a point of quadratic tangency,
it is a parabolic period-$2p$ orbit.
\end{thm}

\begin{proof}
Since $\xi \in \Fix{S}\cap f^p(\Fix{S})$, then $\xi = S(\xi)$ and there is a point $\eta \in 
\Fix{S}$ such that $f^p(\eta)= \xi$.
Consider first $p=1$. Then we have $f^{2}(\eta)= f(f(\eta))=$ $f(\xi)= f(S(\xi))= S(f^{-1}
(\xi))$ $=S(\eta)=\eta$. Similarly, one has $f^2(\xi)=\xi$. By induction, the same is true for any $p\in \bZ$. Below we work with $p=1$ to facilitate calculations.

According to the Bochner-Montgomery theorem \cite{Bochner46} we can take two symplectic 
charts: $\cV$ near $\eta$ with coordinates $(x,y)$ and $\cU$ near $\xi$ with coordinates $
(u,v)$ such that in $\cV$ the involution $S$ becomes $S(x,y)= (x,-y)$, and similarly in $\cU$ 
it becomes $S(u,v)= (u,-v)$. Moreover, $f|_{\cV}=f_1: \cV \to \cU$ is written as follows (we 
assume with no loss of generality that $\xi$ and $\eta$ have zero coordinates in the related 
charts)
\[
	\begin{pmatrix} u \\ v \end{pmatrix} =
		A \begin{pmatrix} x \\ y \end{pmatrix} +
			\begin{pmatrix} F_1(x,y) \\ G_1(x,y) \end{pmatrix}
%	u=a_{11}x+a_{12}y + F_1(x,y),\;v=a_{21}x+a_{22}y + G_1(x,y),\;du\wedge dv = dx\wedge dy,
\]
where $A$ is a constant matrix and $F_1$ and $G_1$ are $\cO(2)$.
Similarly $f|_{\cU} = f_2: \cU \to \cV$ has the form
\[
	\begin{pmatrix} x \\ y \end{pmatrix} =
		B\begin{pmatrix} u \\ v \end{pmatrix} +
			\begin{pmatrix} F_2(u,v) \\ G_2(u,v) \end{pmatrix} .
%	x=b_{11}u+b_{12}v + F_2(u,v),\;y=b_{21}u+b_{22}v + G_2(u,v),\;du\wedge dv = dx\wedge dy.
\]
Note that in both cases, $du \wedge dv = dx \wedge dy$ by area preservation.

If $\xi$ is the point of transverse intersection of $f_1(\Fix{S})$ and
$\Fix{S}$, then two vectors $(a_{11},a_{21})^\top$ and $(1,0)^\top$ are
transverse, i.e., $a_{21}\ne 0$. In this case, when $0<a_{12}a_{21}<1$, the point $\eta$ is
elliptic (its eigenvalues satisfy $|\lambda_{1,2}|=1$), while if $a_{12}a_{21} < 0$ it is an 
orientable saddle, and if $a_{12}a_{21} >1$ it is a non-orientable saddle.

The tangency of $Df_1(\Fix{S})$ and $\Fix{S}$ at $\xi$ implies $a_{21}=0$ and
area preservation gives $a_{22}=a^{-1}_{11}$. The reversibility written in both
coordinate charts provides the following relations for direct and inverse maps
$f_1\circ S = S\circ f_2^{-1}$, $f_2\circ S = S\circ f_1^{-1}$, or in
coordinate form:
\[
	f_1^{-1}: \begin{pmatrix} x \\ y \end{pmatrix} =
		\begin{pmatrix} a_{22} & -a_{12} \\ 0 & a_{11} \end{pmatrix}
			\begin{pmatrix} u \\ v \end{pmatrix} +
		  \begin{pmatrix} F_2(u,-v) \\ - G_2(u,-v) \end{pmatrix},
\]
and
\[
	 f_2^{-1}: \begin{pmatrix} u \\ v \end{pmatrix} =
		\begin{pmatrix} a_{11} & a_{12} \\ a_{21} & a_{22} \end{pmatrix}
			\begin{pmatrix} x \\ y \end{pmatrix} +
		  \begin{pmatrix} F_2(x,y) \\ G_2(x,y) \end{pmatrix},	
%	 u=a_{11}x+a_{12}y + F_1(x,y),\;v=a_{21}x+a_{22}y + G_1(x,y),\;du\wedge dv = dx\wedge dy,
\]
from where we get relations: $a_{11}=b_{22}$, $a_{12}=b_{12}$, $a_{22}=b_{11}$, $b_{21}=0$,
$U_2(x,y)=F_1(x,-y)$, $V_2(x,y)=-G_1(x,-y)$, $U_1(u,v)=F_2(u,-v)$,
$V_1(u,v)=-G_2(u,-v)$, here $U_1, V_1$, $U_2, V_2$ are nonlinear terms of
the inverse maps $f_1^{-1}, f_2^{-1}$. Denote below for brevity
$a_{11}=\alpha$, $a_{12}=\beta$, then $a_{22}=\alpha^{-1}$.

The quadratic tangency of $f_1(\Fix{S})$ and
$\Fix{S}$ at $\xi$ implies $\partial^2 G_1/\partial x^2 \ne 0$ at $(0,0)$.
The map $f^2$ near a 2-periodic point $\eta$ has the form $f_2\circ f_1$.
Hence, the linear part of this map has the matrix
\[
\begin{pmatrix}1&2\beta/\alpha \\ 0&1\end{pmatrix},\; \gamma = 2\beta/\alpha \ne 0.
\]
Let us notice that for the map $f^2$ near the point $\eta$ to guarantee
its fixed point be parabolic (not more higher degeneration) we
need only to check that in the local coordinates
\[
x_1 = x+\gamma y + p(x,y),\; y_1 = y + q(x,y),\;dx_1\wedge dy_1 = dx\wedge dy
\]
the inequality $\partial^2 q/\partial x^2 \ne 0$ at the fixed point holds.
For our case this quantity is the following
\[
\frac{\partial^2 q}{\partial x^2}(0,0)=\alpha \frac{\partial^2 G_1}{\partial x^2}(0,0)-
\alpha^2 \frac{\partial^2 V_1}{\partial u^2}(0,0).
\]
From identities derived from the representation for $f_1$ and $f_2 = S\circ f_1^{-1}\circ
S$ we get
\[
\frac{\partial^2 V}{\partial u^2}(0,0)=-\frac{1}{\alpha}\frac{\partial^2 G_1}{\partial x^2}
(0,0),
\]
therefore we come to
\[
\frac{\partial^2 q}{\partial x^2}(0,0)=\alpha \frac{\partial^2 G_1}{\partial x^2}(0,0)-
\alpha^2 \frac{\partial^2 V_1}{\partial u^2}(0,0)=2\alpha
\frac{\partial^2 G_1}{\partial x^2}(0,0)\ne 0
\]
due to the quadratic tangency of $\Fix{S}$ and $f(\Fix{S})$ at $\xi$.
\end{proof}

%%%%%%%%%%%%%%%%
% Bounds
%%%%%%%%%%%%%%%%
\section{Orbit Bounds}\label{app:UnBounded}
In this appendix, we obtain a sufficient condition for the forward orbit of a point under the 
lift $F$ of the map \Eq{StdMap} to be unbounded. This condition is used in the proof of \Th{Slope}.

Write the lift as
\[
	\TwoVec{x_{t+1}}{y_{t+1}} =
	 A \TwoVec{x_t}{y_t} + \hat g(x_t) \TwoVec{1}{1} ,
\]
where $A$ is the matrix in \Eq{Anosov}, and $\hat g(x+1) = \hat g(x)$.
The formal solution to this iteration is
\beq{formalSol}
	\TwoVec{x_t}{y_t} = A^t \TwoVec{x_0}{y_0} + \sum_{j=0}^{t-1} \hat g(x_{t-1-k}) A^j \TwoVec{1}{1} .
\eeq
The $t^{th}$ power of the Anosov matrix \Eq{Anosov} is easily computed in terms of
the Fibonacci sequence,
\beq{Fibonacci}
	F_{t+1} = F_{t} + F_{t-1}, \quad F_{-2} = 1, \quad F_{-1} = 0 ,
\eeq
to obtain
\[
	A^t = \begin{pmatrix} F_{2t} & F_{2t-1} \\ F_{2t-1} & F_{2t-2} \end{pmatrix} .
\]
Thus \Eq{formalSol} becomes
\[
	\TwoVec{x_t}{y_t} = A^t \TwoVec{x_0}{y_0} + \sum_{j=0}^{t-1} \hat g(x_{t-1-j}) \TwoVec{F_{2j+1}}{F_{2j}} .
\]	
Supposing that $\|\hat g(x)\|_\infty = G$, we can find a lower bound on the orbit as
\begin{align*}
		x_t &\ge F_{2t} x_0 + F_{2t-1}y_0 - G \sum_{j=0}^{t-1} F_{2j+1} ,\\
		y_t &\ge F_{2t-1} x_0 + F_{2t-2}y_0 - G \sum_{j=0}^{t-1} F_{2j} .
\end{align*}

The solution to the Fibonacci difference equation \Eq{Fibonacci} is
\[
	F_t = \frac{\phi+2}{5} \left[ \phi^t + (-\phi)^{-t-2}\right] \ge \frac{\phi+2}{5}(\phi^t-1) ,
\]
where $\phi$ is the golden mean. Thus
\begin{align*}
	\sum_{j=0}^{t-1} F_{2j+1} &= \frac{\phi+2}{5}\left[ \phi^{2t}-1 +\phi^{-2}(\phi^{-2t}-1) \right] \le \frac{\phi+2}{5} \phi^{2t} ,\\
	\sum_{j=0}^{t-1} F_{2j} &= \frac{\phi+2}{5}\left[\phi^{2t-1}- \phi^{-2t-1}\right]
	   \le \frac{\phi+2}{5} \phi^{2t-1} .
\end{align*}
Consequently if $x_0, y_0 \ge 0$, then
\begin{align*}
		x_t &\ge \frac{\phi+2}{5}\left [(\phi^{2t}(x_0 + \phi^{-1} y_0 - G) -x_0 -y_0 \right], \\
		y_t &\ge \frac{\phi+2}{5}\left[(\phi^{2t-1}(x_0 + \phi^{-1}y_0 - G) -x_0 -y_0\right] .
\end{align*}
Therefore, whenever
\beq{Initial bounds}
  x_0 + \phi^{-1} y_0 > G, \quad x_0,y_0 > 0 ,
\eeq
then we have $x_t, y_t \to \infty$ as $t \to \infty$.

For the form \Eq{Force} with $\mu = \mu_p(\eps)$ from \Eq{Mu_p}, the sup-norm
of $\hat g$ is
\[
	G = \tfrac{1}{2\pi}[\sqrt{\eps(1+\eps)} + 1+\eps - \sec^{-1}(1+\eps)]
	 \le \tfrac{1}{2\pi}(2\eps +1) .
\]
Thus the forward orbit of a point $(x_0,y_0)$ in the positive quadrant that satisfies
\beq{FinalBounds}
  x_0 + \phi^{-1} y_0 > \tfrac{1}{2\pi}(2\eps +1)
\eeq
is unbounded.

%%%%%%%%%%%%%%%%%
%% Area by Action
%%%%%%%%%%%%%%%%%
\section{Actions and Areas}\label{app:Action}
Areas bounded by segments of invariant manifolds of an exact, area-preserving map $F:\bR^2\to
\bR^2$ can be computed using the action-flux formulas of MacKay, Meiss, and Percival
\cite{MMP84, MMP87, Meiss92}. In particular, suppose that $F$ preserves the area form $\omega
$,
i.e., $F^*\omega = \omega$, and $\omega = -d\nu$ is an exact form. We say that $F$ is exact,
area-preserving when there exists a zero-form $L:M \to \bR$ such that
\beq{Exact}
	F^*\nu - \nu = dL
\eeq
In particular, the lift of \Eq{StdMap} is exact symplectic with form $\omega = dx \wedge dy$
with the Lagrangian
\beq{Lagrangian}
	L(x,y) = \frac12 (y+g(x))^2 + G(x)
\eeq
where $G$ is any anti-derivative of $g$.

Suppose that $z^* = (x^*,y^*)$ is a hyperbolic or parabolic fixed point of $F$ and
$\cU = W^u(z^*,\zeta)$ is the segment of the right-going unstable manifold between $z^*$
and the point $\zeta = (\xi,\eta) \in W^u(z^*)$. Let $ \zeta_t = F^{t}(\zeta)$
denote points on the orbit of $\zeta = \zeta_0$, so that $\zeta_t \to z^*$ as $t \to -\infty$.

Consider the region $\cR$ ``below" the segment $\cU$ and above $x$-axis, as sketched in
\Fig{ChannelsE10}. This region is bounded by the loop
\[
	\partial \cR = \{ (x,0): x^*\le x \le \xi\} + \{(\xi,y): 0\le y \le \eta\} -\cU
	 - \{(x^*,y): 0<y<y^*\} .
\]
The area of $\cR$ is
\[
	 A^u = \int_{\cR} \omega = -\oint_{\partial \cR} \nu = \int_{\cU} \nu ,
\]
upon doing the trivial integrals along the straight segments of $\partial \cR$. The remaining
integral of the one-form $\nu$ along the segment $\cU$ can be done using \Eq{Exact},
recursion,
and the fact that $F^{-t}(\cU) \to z^*$:
\bsplit{Action}	
	 \int_{\cU} \nu &= \int_{F^{-1}(\cU)} dL + \int_{F^{-1}(\cU)} \nu
	 = L(\zeta_{-1}) -L(z^*) + \int_{F^{-1}(\cU)} \nu\\
	 &= \sum_{t=-\infty}^{-1} ( L(\zeta_t)-L(z^*))
	 \equiv \Delta \cA^-(\zeta,z^*) ,
\esplit
the difference between the past actions of the two orbits.

For the map \Eq{StdMap}, the upper half of the channel $\tilde E^+(\xi)$ has boundary
\Eq{ChannelBoundary}. Since the fixed points have $y^* = 0$, the channel area is the
difference
between the areas below the hyperbolic manifold and that below the parabolic manifold, as
given
by \Eq{Action}:
\[
	A^+(\xi) = Area(E^+(\xi)) = \Delta \cA^-(\zeta_h,h) - \Delta \cA^-(\zeta_p,p).
\]
Areas computed using this formula for the map \Eq{StdMap} are shown in
\Fig{LiftedChannelArea}.
%%%%%
\InsertFig{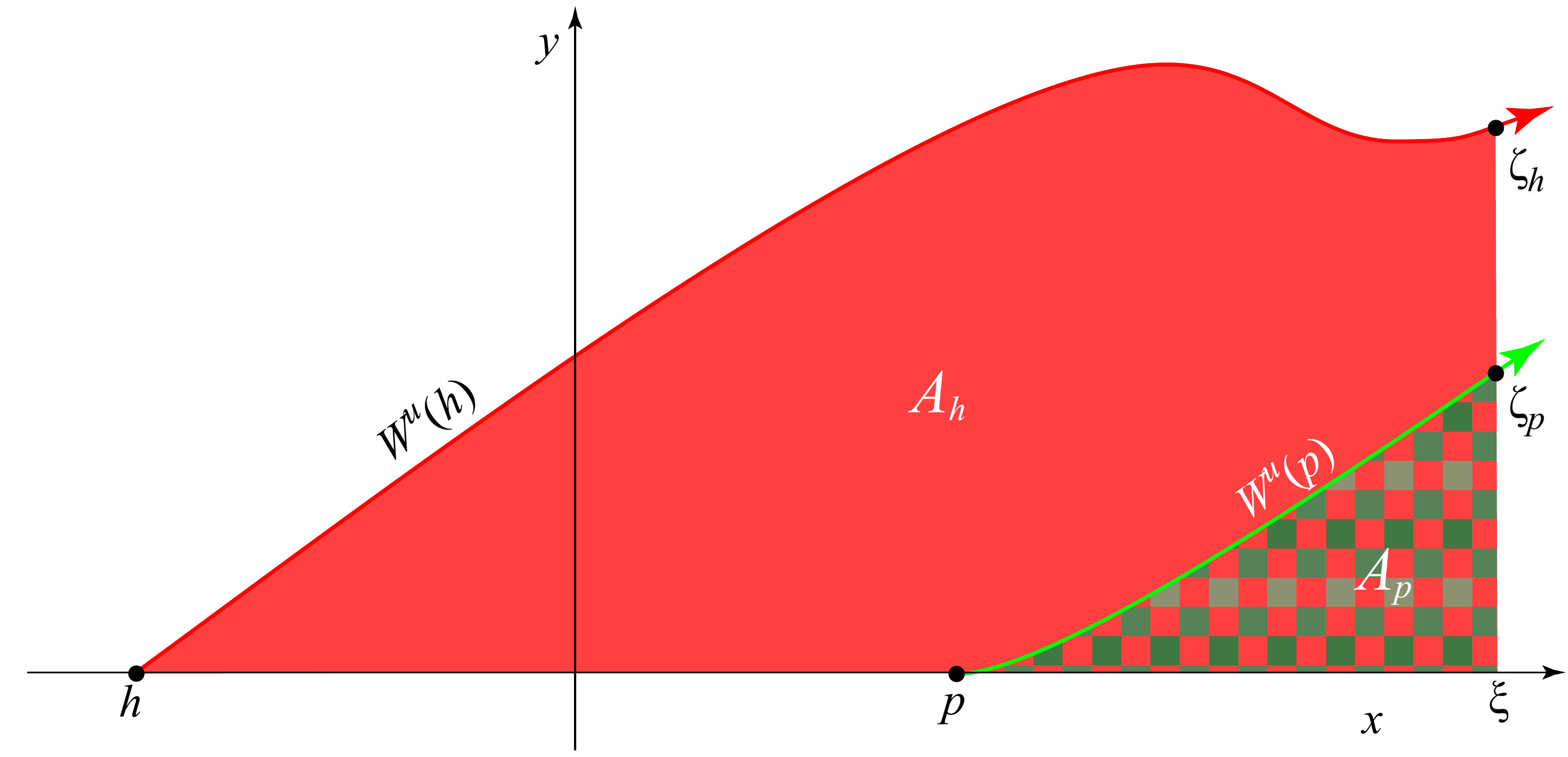}{Areas below the initial segments of the right-going unstable manifolds
of $(x_h,0)$ (red), and $(x_p,0)$
(green checkerboard) for $\eps = 1.0$.}{ChannelsE10}{3in}
%%%%%

For a symmetric fixed point $z^*$, the reversor, \Eq{Reversor}, maps
$S(W^u(z^*)) = W^s(z^*)$. The image of the channel $\tilde E^+$ is bounded by the curve
\[
	S(\partial \tilde E^+) = \{ (x,0): x_h\le x \le \xi\} + W^s(p,S(\zeta_p))
	 + \{(\xi-y,y): -\eta_p \ge y \ge -\eta_h\} -W^s(h, S(\zeta_h))
\]
Note that the reflected channel has a cut-off that is a line segment with slope minus one.
Now since $S$ is area-preserving, but orientation reversing, the area of the stable channel
$\tilde E^- = S(\tilde E^+)$ is
\[
	\int_{\tilde E^-} \omega = \oint_{\partial{\tilde E^-}} \nu
	  = -\oint_{S(\partial \tilde E^+)} \nu = -\int_{\tilde E^+} \omega .
\]
Thus, up to the sign, the areas are the same.

%%%%%%%%%%%%%%%
\bibliographystyle{alpha}
\bibliography{Parabolic}

\begin{thebibliography}{MMP87}

\bibitem[AA90]{Aubry90}
S.~Aubry and G.~Abramovici.
\newblock Chaotic trajectories in the standard map, the concept of
  anti-integrability.
\newblock {\em Physica D}, 43:199--219, 1990.
\newblock \url{http://dx.doi.org/10.1016/0167-2789(90)90133-A}.

\bibitem[ALD83]{Aubry83b}
S.~Aubry and P.Y. Le~Daeron.
\newblock The discrete {F}renkel-{K}ontorova model and its extensions.
\newblock {\em Physica D}, 8:381--422, 1983.
\newblock \url{http://dx.doi.org/10.1016/0167-2789(83)90233-6}.

\bibitem[Arn63]{Arnold63}
V.I. Arnold.
\newblock Small denominators and problems of stability of motion in classical
  and celestial mechanics.
\newblock {\em Russ. Math. Surveys}, 18:6:85--191, 1963.
\newblock \url{http://dx.doi.org/10.1070/RM1963v018n06ABEH001143}.

\bibitem[Ber78]{Berry78}
M.V. Berry.
\newblock Regular and irregular motion.
\newblock In S.~Jorna, editor, {\em Topics in Nonlinear Dynamics : A Tribute to
  Sir Edward Bullard}, volume~46 of {\em AIP Conf. Proc.}, pages 16--120. AIP,
  New York, 1978.
\newblock \url{http://dx.doi.org/10.1063/1.31417}.

\bibitem[BM46]{Bochner46}
S.~Bochner and D.~Montgomery.
\newblock Locally compact groups of differentiable transformations.
\newblock {\em Ann. of Math. (2)}, 47(4):639--653, 1946.
\newblock \url{http://www.jstor.org/stable/1969226}.

\bibitem[Boc02]{Bochi02}
J.~Bochi.
\newblock Genericity of zero {L}yapunov exponents.
\newblock {\em Ergod. Theory Dynam. Syst.}, 22:1667--1696, 2002.
\newblock \url{http://dx.doi.org/10.1017/S0143385702001165}.

\bibitem[CE01]{Enrich01}
E.~Catsigeras and H.~Enrich.
\newblock {SRB} measures of certain almost hyperbolic diffeomorphisms with
  tangency.
\newblock {\em Discr. Cont. Dynam. Systet., ser. A}, 7(1):177--202, 2001.
\newblock \url{http://dx.doi.org/10.3934/dcds.2001.7.177}.

\bibitem[Dua08]{Duarte08}
P.~Duarte.
\newblock Elliptic isles in families of area-preserving maps.
\newblock {\em Ergodic Theory and Dynamical Systems}, 28(06):1781--1813, 2008.
\newblock \url{http://dx.doi.org/10.1017/S0143385707000983}.

\bibitem[Fon99]{Fontich99}
E.~Fontich.
\newblock Stable curves asymptotic to a degenerate fixed point.
\newblock {\em Nonlinear Anal: Theory, Methods, Appl.}, 35:711--733, 1999.
\newblock \url{http://dx.doi.org/10.1016/S0362-546X(98)00004-2}.

\bibitem[Fra70]{Franks70}
J.~Franks.
\newblock {A}nosov diffeomorphisms.
\newblock In S-S Chern and S.~Smale, editors, {\em 1970 Global Analysis},
  volume XIV of {\em Proc Symp. on Pure Math}, pages 61--93. Am. Math. Soc,
  1970.

\bibitem[GK82]{Gerber82}
M.~Gerber and A.~Katok.
\newblock Smooth models of {T}hurston's pseudo-{A}nosov maps.
\newblock {\em Ann. Scient. \'{E}c. Norm. Super., ser.4}, 15:173--204, 1982.
\newblock \url{http://www.numdam.org/item?id=ASENS_1982_4_15_1_173_0}.

\bibitem[Gri77]{Grines77}
V.Z. Grines.
\newblock The topological conjugacy of diffeomorphisms of a two-dimensional
  manifold on one-dimensional orientable basic sets. {II}.
\newblock {\em Trudy Moskovskogo Matematicheskogo Obshchestva}, 34:243--252,
  1977.
\newblock \url{http://mi.mathnet.ru/mmo338}.

\bibitem[GTS07]{GTSh07}
S.~Gonchenko, D.~Turaev, and L.~Shilnikov.
\newblock Homoclinic tangencies of arbitrarily high orders in conservative and
  dissipative two-dimensional maps.
\newblock {\em Nonlinearity}, 20:241--275, 2007.
\newblock \url{http://dx.doi.org/10.1088/0951-7715/20/2/002}.

\bibitem[HPS77]{Hirsch77}
M.W. Hirsch, C.~Pugh, and M.~Shub.
\newblock {\em Invariant Manifolds}, volume 583 of {\em Lecture Notes in
  Mathematics}.
\newblock Springer-Verlag, New York, 1977.

\bibitem[Kat79]{Katok79}
A.K. Katok.
\newblock Bernoulli diffeomorphisms on surfaces.
\newblock {\em Ann. Math.}, 110:529--547, 1979.
\newblock \url{http://dx.doi.org/10.2307/1971237}.

\bibitem[Ler10]{Lerman10}
L.~Lerman.
\newblock Breaking hyperbolicity for smooth symplectic toral diffeomorphisms.
\newblock {\em Regul. Chaotic Dyn.}, 15(2-3):194--209, 2010.
\newblock \url{http://dx.doi.org/10.1134/S1560354710020085}.

\bibitem[Lew80]{Lewowicz80}
J.~Lewowicz.
\newblock Lyapunov functions and topological stability.
\newblock {\em J. Differential Equations}, 38:192--209, 1980.
\newblock \url{http://dx.doi.org/10.1016/0022-0396(80)90004-2}.

\bibitem[Liv04]{Liverani04}
C.~Liverani.
\newblock Birth of an elliptic island in a chaotic sea.
\newblock {\em Math. Phys. Electron. J.}, 10:1--13, 2004.
\newblock \url{http://www.maia.ub.es/mpej/Vol/10/1.pdf}.

\bibitem[LR98]{Lamb98}
J.W.S. Lamb and J.A.G. Roberts.
\newblock Time-reversal symmetry in dynamical systems: A survey.
\newblock {\em Physica D}, 112:1--39, 1998.
\newblock \url{http://dx.doi.org/10.1016/S0167-2789(97)00199-1}.

\bibitem[Mac93]{MacKay93}
R.S. MacKay.
\newblock {\em Renormalisation in Area-Preserving Maps}, volume~6 of {\em Adv.
  Series in Nonlinear Dynamics}.
\newblock World Scientific, Singapore, 1993.

\bibitem[Mar69]{Markley69}
N.G. Markley.
\newblock The {P}oincar\'e-{B}endixon theorem for the {K}lein bottle,.
\newblock {\em Trans. A.M.S.}, 135:159--165, 1969.
\newblock \url{http://dx.doi.org/10.2307/1995009}.

\bibitem[Mat82]{Mather82}
J.N. Mather.
\newblock Existence of quasi-periodic orbits for twist homeomorphisms of the
  annulus.
\newblock {\em Topology}, 21:457--467, 1982.
\newblock \url{http://dx.doi.org/10.1016/0040-9383(82)90023-4}.

\bibitem[Mei92]{Meiss92}
J.D. Meiss.
\newblock Symplectic maps, variational principles, and transport.
\newblock {\em Reviews of Modern Physics}, 64(3):795--848, 1992.
\newblock \url{http://dx.doi.org/10.1103/RevModPhys.64.795}.

\bibitem[MMP84]{MMP84}
R.S. MacKay, J.D. Meiss, and I.C. Percival.
\newblock Transport in {H}amiltonian systems.
\newblock {\em Physica D}, 13:55--81, 1984.
\newblock \url{http://dx.doi.org/10.1016/0167-2789(84)90270-7}.

\bibitem[MMP87]{MMP87}
R.S. MacKay, J.D. Meiss, and I.~C. Percival.
\newblock Resonances in area-preserving maps.
\newblock {\em Physica D}, 27(1-2):1--20, 1987.
\newblock \url{http://dx.doi.org/10.1016/0167-2789(87)90002-9}.

\bibitem[New77]{NewHouse77}
S.E. Newhouse.
\newblock Quasi-elliptic periodic points in conservative dynamical systems.
\newblock {\em Am. J. Math.}, 99:1061--1087, 1977.
\newblock \url{http://dx.doi.org/10.2307/2374000}.

\bibitem[Prz82]{Przytycki82}
F.~Przytycki.
\newblock Examples of conservative diffeomorphisms of the two-dimensional torus
  with coexistence of elliptic and stochastic behaviour.
\newblock {\em Ergod. Th. Dyn. Sys.}, 2:439--463, 1982.
\newblock \url{http://dx.doi.org/10.1017/S0143385700001711}.

\bibitem[Sin95]{Sinai95}
Ya. Sinai.
\newblock A mechanism of ergodicity in the standard map.
\newblock In C.~Simo, editor, {\em Hamiltonian Systems with Three or More
  Degrees of Freedom}, volume 533 of {\em Series C: Mathematical and Physical
  Sciences}, pages 242--243. NATO ASI, 1995.
\newblock \url{http://dx.doi.org/10.1007/978-94-011-4673-9}.

\bibitem[Sma67]{Smale67}
S.~Smale.
\newblock Differentiable dynamical systems.
\newblock {\em Bull. Am. Math. Soc.}, 73:747--817, 1967.
\newblock
  \url{http://www.ams.org/journals/bull/1967-73-06/S0002-9904-1967-11798-1/}.

\bibitem[Wei36]{Weil36}
A.~Weil.
\newblock Les familles de courbes sur le tore.
\newblock {\em Mat. Sborn., Nov. Ser.}, 1(5):779--781, 1936.

\end{thebibliography}

\end{document}